\documentclass[letterpaper,11pt]{article}

\usepackage[margin=1in]{geometry}  

\usepackage{enumerate}
\usepackage[fixpdftex]{xcolor}
\usepackage[title]{appendix}

\usepackage{amsmath} 
\usepackage{amssymb}
\usepackage{amsfonts}
\usepackage{amsthm}
\usepackage{mathrsfs}
\usepackage{color, xcolor}
\usepackage{prettyref}
\usepackage{indentfirst}
\usepackage{graphicx}
\usepackage{xspace}
\usepackage{subfig}

\usepackage{ifthen}
\newboolean{fullversion}
\setboolean{fullversion}{TRUE} 
\newboolean{slides}
\setboolean{slides}{FALSE}		
\newboolean{conf}
\setboolean{conf}{FALSE} 

\usepackage{stmaryrd}

\usepackage{empheq}

\definecolor{myblue}{rgb}{.8, .8, 1}
\definecolor{mathblue}{rgb}{0.2472, 0.24, 0.6} 
\definecolor{mathred}{rgb}{0.6, 0.24, 0.442893}
\definecolor{mathyellow}{rgb}{0.6, 0.547014, 0.24}

\usepackage{soul}
\definecolor{lightblue}{rgb}{.90,.95,1}
\sethlcolor{lightblue}

\usepackage[
            CJKbookmarks=true,
            bookmarksnumbered=true,
            bookmarksopen=true,
            colorlinks=true,
            citecolor=red,
            linkcolor=blue,
            anchorcolor=red,
            urlcolor=blue,
            pdfauthor={Yihong Wu}
            ]{hyperref}

\usepackage{pgf}
\usepackage{authblk}
\usepackage[percent]{overpic}
\usepackage{rotating}

\usepackage{tikz}
\usetikzlibrary{arrows,calc}
\tikzstyle{int}=[draw, fill=blue!20, minimum size=2em]
\tikzstyle{dot}=[circle, draw, fill=blue!20, minimum size=2em]
\tikzstyle{init} = [pin edge={to-,thin,black}]

\usepackage{pgfplots}

\pgfplotsset{
    standard/.style={
        axis x line=middle,
        axis y line=middle,
        every axis x label/.style={at={(current axis.right of origin)},anchor=north west},
        every axis y label/.style={at={(current axis.above origin)},anchor=north west}
    }
}

\newcommand{\betau}{\overline{\beta}^*}
\newcommand{\betal}{\underline{\beta}^*}
\newcommand{\betas}{\beta^{\sharp}}

\newcommand{\HC}{\mathrm{HC}}
\newcommand{\Ray}{\mathrm{Ray}}

\newcommand{\eg}{e.g.\xspace}
\newcommand{\ie}{i.e.\xspace}
\newcommand{\iid}{i.i.d.\xspace}

\newcommand{\reals}{{\mathbb{R}}}

\newcommand{\naturals}{{\mathbb{N}}}

\newcommand{\eexp}{{\rm e}}

\newcommand{\diff}{{\rm d}}

\newcommand{\expect}[1]{\mathbb{E}\left[ #1 \right]}

\newcommand{\Prob}{\mathbb{P}}
\newcommand{\prob}[1]{{ \mathbb{P}\left\{ #1 \right\} }}

\newcommand{\toprob}{\xrightarrow{\Prob}}

\newcommand{\eqdistr}{{\stackrel{\rm (d)}{=}}}
\newcommand{\iiddistr}{{\stackrel{\text{\iid}}{\sim}}}
\newcommand{\var}{\mathsf{var}}

\newcommand{\Bern}{\text{Bern}}

\newcommand{\diverge}{\to \infty}

\ifthenelse{\boolean{slides}}{
\theoremstyle{remark}
\newtheorem{remark}{Remark}}{
\theoremstyle{plain}

\newtheorem{lemma}{Lemma}
\newtheorem{theorem}{Theorem}

\theoremstyle{definition}
\newtheorem{definition}{Definition}

\newtheorem{remark}{Remark}

}

\theoremstyle{plain}

\newtheorem{coro}{Corollary}

\newrefformat{eq}{(\ref{#1})}
\newrefformat{chap}{Chapter~\ref{#1}}
\newrefformat{sec}{Section~\ref{#1}}
\newrefformat{algo}{Algorithm~\ref{#1}}
\newrefformat{fig}{Fig.~\ref{#1}}
\newrefformat{tab}{Table~\ref{#1}}
\newrefformat{rmk}{Remark~\ref{#1}}
\newrefformat{clm}{Claim~\ref{#1}}
\newrefformat{def}{Definition~\ref{#1}}
\newrefformat{cor}{Corollary~\ref{#1}}
\newrefformat{lmm}{Lemma~\ref{#1}}
\newrefformat{prop}{Proposition~\ref{#1}}
\newrefformat{app}{Appendix~\ref{#1}}
\newrefformat{ex}{Example~\ref{#1}}
\newrefformat{exer}{Exercise~\ref{#1}}
\newrefformat{soln}{Solution~\ref{#1}}

\newcommand{\lunder}[1]{{\underset{\raise0.3em\hbox{$\smash{\scriptscriptstyle-}$}}{#1}}}

\newcommand{\floor}[1]{{\left\lfloor {#1} \right \rfloor}}
\newcommand{\ceil}[1]{{\left\lceil {#1} \right \rceil}}

\newcommand{\linf}[1]{\left\|{#1} \right\|_{\infty}}

\newcommand{\esssup}{\mathop{\mathrm{ess \, sup}}}
\newcommand{\essinf}{\mathop{\mathrm{ess \, inf}}}

\newcommand{\indc}[1]{{\mathbf{1}_{\left\{{#1}\right\}}}}
\newcommand{\Indc}{\mathbf{1}}

\def\innergetnumber#1[#2]#3{#2}
\def\getnumber{\expandafter\innergetnumber\jobname}

\newcommand{\bbF}{{\mathbb{F}}}

\newcommand{\sfE}{{\mathsf{E}}}

\newcommand{\calE}{{\mathcal{E}}}
\newcommand{\calF}{{\mathcal{F}}}

\newcommand{\calN}{{\mathcal{N}}}

\newcommand{\tG}{{\tilde{G}}}

\newcommand{\comp}[1]{{#1^{\rm c}}}

\newcommand{\ntok}[2]{{#1,\ldots,#2}}

\newcommand{\renyi}{R\'enyi\xspace}

\xdefinecolor{darkgreen}{rgb}{0,0.35,0}

\newcommand{\pth}[1]{\left( #1 \right)}

\newcommand{\sth}[1]{\left\{ #1 \right\}}

\newcommand{\fracd}[2]{\frac{\diff #1}{\diff #2}}

\newcommand{\TV}{{\sf TV}}

\newcommand{\xxx}{\textbf{xxx}\xspace}

\author{T. Tony Cai\thanks{\url{tcai@wharton.upenn.edu}} } 
\author{Yihong Wu\thanks{\url{yihongwu@wharton.upenn.edu}}}
\affil{Department of Statistics \\
The Wharton School\\
University of Pennsylvania\\
Philadelphia, PA 19104, USA}

\title{Optimal Detection For Sparse Mixtures\footnote{The research  was supported in part
by NSF FRG   Grant DMS-0854973.}}
\date{\today}

\begin{document}
\DeclareGraphicsExtensions{.pgf}
\graphicspath{{figures/}}

\maketitle

\let\oldthefootnote\thefootnote
\renewcommand{\thefootnote}{\fnsymbol{footnote}}
\let\thefootnote\oldthefootnote

\begin{abstract}
Detection of sparse signals arises in a wide range of modern scientific studies. The focus so far has been mainly on Gaussian mixture models. In this paper, we consider the detection problem under a general sparse mixture model and obtain an explicit expression for the detection boundary. It is shown that the fundamental limits of detection is governed by the behavior of the log-likelihood ratio evaluated at an appropriate quantile of the null distribution. We also establish the adaptive optimality of the higher criticism procedure across all sparse mixtures satisfying certain mild regularity conditions. In particular, the general results obtained in this paper recover and extend  in a unified manner the previously known results on sparse detection far beyond the conventional Gaussian model and other exponential families. 

\end{abstract}

\vspace{1em}

%

\newpage
\section{Introduction}
	\label{sec:into}
			
Detection of sparse  mixtures is an important problem that arises in many scientific applications such as signal processing \cite{Dobrushin58}, biostatistics \cite{JCL10}, and astrophysics \cite{CJT05, JSDAF05}, where the goal is to determine the existence of a signal which only appears in a small fraction of the noisy data. For example,  topological defects and Doppler effects manifest themselves as non-Gaussian convolution component in the Cosmic Microwave Background (CMB) temperature fluctuations. Detection of non-Gaussian signatures are important to identify cosmological origins of many phenomena \cite{JSDAF05}.  Another example is disease surveillance where it is critical to discover an outbreak when the infected population is small \cite{Kulldorff05}. 
The detection problem is of significant interest also because it is closely connected to a number of other important problems including estimation, screening, large-scale multiple testing, and classification. See, for example,  \cite{CJJ11}, \cite{CJL07}, \cite{DJ.HC}, \cite{Hall08}, and \cite{JCL10}.

\subsection{Detection of sparse binary vectors}
	\label{sec:intro.detection}	
	One of the earliest work on sparse mixture detection dates back to Dobrushin \cite{Dobrushin58}, who considered the following problem originating from multi-channel detection in radiolocation. Let $\Ray(\alpha)$ denote the Rayleigh distribution with the density $\frac{2y}{\alpha}\exp(-\frac{y^2}{\alpha}), y \geq 0$. Let $\{Y_i\}_{i=1}^n$ be independently distributed according to $\Ray(\alpha_i)$, representing the random voltages observed on the $n$ channels. In the absence of noise, $\alpha_i$'s are all equal to one, the nominal value; while in the presence of signal, exactly one of the $\alpha_i$'s
	becomes a known value $\alpha > 1$. Denoting the uniform distribution on $[n]$ by $U_n$, the goal is to test the following competing hypotheses
			\begin{equation}
H_{0}^{(n)}:  \alpha_i = 1,  i \in [n], \quad \mbox{versus}\quad
H_{1}^{(n)}:  \alpha_i = 1 + (\alpha -1) \indc{i=J}, \quad
J \sim U_n \, .
	\label{eq:HT.dobrushin}
\end{equation}
Since the signal only appears once out of the $n$ samples, in order for the signal to be distinguishable from noise, it is necessary for the amplitude $\alpha$ to grow with the sample size $n$ (in fact, at least logarithmically). By proving that the log-likelihood ratio converges to a stable distribution in the large-$n$ limit, Dobrushin \cite{Dobrushin58} obtained sharp asymptotics of the smallest $\alpha$ in order to achieve the desired false alarm and miss detection probabilities. Similar results are obtained in the continuous-time Gaussian setting by Burnashev and Begmatov \cite{BB90}.
	
 Subsequent important work include Ingster \cite{Ingster97} and Donoho and Jin  \cite{DJ.HC}, which focused on detecting a sparse binary vector in the presence of Gaussian observation noise. The problem can be formulated as follows.  Given a random sample $\{Y_1, ...,  Y_n\}$, one wishes to test the hypotheses
	\begin{equation}
H_{0}^{(n)}:  Y_i \, \iiddistr \, \calN(0,1),  i \in [n]  \quad \mbox{versus}\quad
H_{1}^{(n)}:  Y_i \, \iiddistr \, (1-\epsilon_n) \calN(0,1) + \epsilon_n \calN(\mu_n, 1),  i \in [n]	
	\label{eq:HT.IDJ}
\end{equation}
where the non-null proportion $\epsilon_n$ is calibrated according to
\begin{equation}
\epsilon_n = n^{-\beta}, \quad \frac{1}{2} < \beta< 1,	
	\label{eq:epsn}
\end{equation}
and the non-null effect $\mu_n$ grows with the sample size according to
\begin{equation}
	\mu_n = \sqrt{2 r \log n}, \quad r > 0.
	\label{eq:r}
\end{equation}
Equivalently, one can  write
\begin{equation}
	Y_i = X_i + Z_i
	\label{eq:conv}
\end{equation}
where $Z_i \iiddistr \calN(0,1)$ is the observation noise. Under the null hypothesis, the mean vector $X^n = (\ntok{X_1}{X_n})$ is equal to zero; under the alternative, $X^n$ is a non-zero sparse binary vector with $X_i \iiddistr (1-\epsilon_n) \delta_0 + \epsilon_n \delta_{\mu_n}$, where $\delta_a$ denotes the point mass at $a$.

\emph{The detection boundary}, which gives the smallest possible signal strength, $r$, such that reliable detection is possible, is given by the following function in terms of the sparsity parameter $\beta$:
	 \begin{equation}
	r^*(\beta)
	= \begin{cases}
	\beta - \frac{1}{2} & \frac{1}{2} < \beta \leq \frac{3}{4} \\
	(1-\sqrt{1-\beta})^2 & \frac{3}{4} < \beta < 1
\end{cases}.
	\label{eq:rstar}
\end{equation}
See  Ingster \cite{Ingster97} and Donoho and Jin  \cite{DJ.HC}.
Therefore, the hypotheses in \prettyref{eq:HT.IDJ} can be tested with vanishing probability of error if and only if the pair $(\beta,r)$ lies in the strict epigraph 
\begin{equation}
\{(\beta, r): r > r^*(\beta)\},	
	\label{eq:detregion}
\end{equation}
 which is called the \emph{detectable region}. Furthermore, because the fraction of the non-zero mean is very small, most tests based on the empirical moments have no power in detection. Donoho and Jin  \cite{DJ.HC} proposed an adaptive testing procedure based on Tukey's higher criticism statistic  and showed that it attains the optimal detection boundary \prettyref{eq:rstar} without requiring the knowledge of the unknown parameters $(\beta,r)$. 
	
	The above results have been generalized  along various directions within the framework of two-component Gaussian mixtures. Jager and Wellner \cite{JW07} proposed a family of goodness-of-fit tests based on the \renyi divergences \cite[p. 554]{renyi61}, including the higher criticism test as a special case, which achieve the optimal detection boundary adaptively. 
The detection boundary with correlated noise was established in \cite{HJ10} which also proposed a modified version of the higher criticism that achieves the corresponding optimal boundary. 
	In a related setup, \cite{CDH05,CCHZ08,CCP11} considered detecting a signal with a known geometric shape in Gaussian noise. Minimax estimation of the non-null proportion $\epsilon_n$ was studied in Cai, Jin and Low \cite{CJL07}.

	The setup of \cite{Ingster97} and \cite{DJ.HC} specifically focuses on the two-point Gaussian mixtures. Although \cite{Ingster97} and \cite{DJ.HC} provide insightful results for sparse signal detection, the setting is highly restrictive and idealized. In particular, it has the limitation that the signal strength must be a constant under the alternative, \ie, the mean vector $X^n$ takes constant value $\mu_n$ on its support. 
	In many applications, the signal itself varies among the non-null portion of the samples. A natural question is the following: What is the detection boundary if $\mu_n$ varies under the alternative, say with a distribution $P_n$? Motivated by these considerations, the following heteroscedastic Gaussian mixture model was considered in Cai, Jeng and Jin \cite{CJJ11}:
		\begin{equation}
H_{0}^{(n)}:  Y_i \, \iiddistr \, \calN(0,1)   \quad \mbox{versus}\quad
H_{1}^{(n)}:  Y_i \, \iiddistr \, (1-\epsilon_n) \calN(0,1) + \epsilon_n \calN(\mu_n,\sigma^2).	
	\label{eq:HT.CJJ}
\end{equation}
In this case,  \cite[Theorems 2.1 and 2.2]{CJJ11} showed that reliable detection is possible if and only if $r > r^*(\beta,\sigma^2)$ where $r^*(\beta,\sigma^2)$ is given by
\begin{equation}
	r^*(\beta,\sigma^2)
	= \begin{cases}
	(2-\sigma^2)(\beta - \frac{1}{2}) & \frac{1}{2} < \beta \leq 1-\frac{\sigma^2}{4}, \sigma^2  < 2 \\
	(1-\sigma \sqrt{1-\beta})_+^2 & \text{otherwise}
\end{cases}.
	\label{eq:cjjr}
\end{equation}
where $x_+ \triangleq \max(x, 0)$. 
It was also shown that the optimal detection boundary can be achieved by a double-sided version of the higher criticism test.

\subsection{Detection of general sparse mixture}
	\label{sec:intro.general}	

Although the setup in Cai, Jeng and Jin \cite{CJJ11} is more general than that considered in \cite{Ingster97} and \cite{DJ.HC}, it is still restricted to the two-component Gaussian mixtures.
In many applications such as the aforementioned multi-channel detection \cite{Dobrushin58} and astrophysical problems \cite{JSDAF05}, the sparse signal may not be binary and the distribution may not be Gaussian. 
In the present paper, we consider the problem of sparse mixture detection in a general framework where the distributions are not necessarily Gaussian and the non-null effects are not necessarily a  binary vector. More specifically, given a random sample $Y^n=\{Y_1, ...,  Y_n\}$,  we wish to test the following hypotheses 
\begin{equation}
H_{0}^{(n)}:  Y_i \, \iiddistr \, Q_n  \quad \mbox{versus}\quad
H_{1}^{(n)}:  Y_i \, \iiddistr \, (1-\epsilon_n) Q_n + \epsilon_n G_n	
	\label{eq:HT.nG}
\end{equation}
where $Q_n$ is the null distribution and $G_n$ is a distribution modeling the statistical variations of the non-null effects. The non-null proportion $\epsilon_n \in (0,1)$ is calibrated according to \prettyref{eq:epsn}. 

In this paper we obtain an explicit formula for the fundamental limit of the general testing problem \prettyref{eq:HT.nG} under mild technical conditions on the mixture. We also establish the adaptive optimality of the higher criticism procedure across all sparse mixtures satisfying certain mild regularity conditions. In particular, the general results obtained in this paper recover and extend all the previously known results mentioned earlier in a unified manner. The results also generalize the optimality and adaptivity of the higher criticism procedure far beyond the original equal-signal-strength Gaussian setup in \cite{Ingster97,DJ.HC} and the heteroscedastic extension in \cite{CJJ11}. In the most general case, it turns out that the detectability of the sparse mixture is governed by the behavior of the log-likelihood ratio evaluated at an appropriate quantile of the null distribution.

Although our general approach does not rely on the Gaussianity of the model, it is however instructive to begin by considering  the special case of \emph{sparse normal mixture} with $Q_n = \calN(0,1)$, \ie,
\begin{equation}
\begin{cases}
H_{0}^{(n)}: & Y_i \, \iiddistr \, \calN(0,1) \\
H_{1}^{(n)}: & Y_i \, \iiddistr \, (1-\epsilon_n) \calN(0,1) + \epsilon_n G_n
\end{cases}.	
	\label{eq:HT}
\end{equation}
It is of special interest to consider the \emph{convolution model}, where
\begin{equation}
G_n = P_n * \calN(0,1)
	\label{eq:convmodel}
\end{equation}
is a standard normal mixture and $*$ denotes the convolution of two distributions. In this case the hypotheses \prettyref{eq:HT} can be equivalently expressed via the additive-noise model \prettyref{eq:conv},
 where
  $X_i=0$ under the null and $X_i \iiddistr (1-\epsilon_n) \delta_0 + \epsilon_n P_n$ under the alternative. Based on the noisy observation $Y^n$,  the goal is to determine whether $X^n$ is the zero vector or a \emph{sparse} vector, whose support size is approximately $n \epsilon_n$ and non-zero entries are distributed according to $P_n$. Therefore, the distribution $P_n$ represents the prior knowledge of the signal. The case of $P_n$ being a point mass is treated in \cite{Ingster97,DJ.HC}. The case of Rademacher $P_n$ in covered in \cite[Chapter 8]{IS03}. The heteroscedastic case where $P_n$ is Gaussian is considered in \cite{CJJ11}.
These results can be recovered by particularizing the general conclusion in the present paper.

Moreover, our results also shed light on what governs detectability in Gaussian noise when the signal does not necessarily have equal strength. For example, consider the classical setup \prettyref{eq:HT.IDJ} where the signal strength $\mu_n$ is now a random variable. If we have $\mu_n = \sqrt{2r \log n} \, X$ for some random variable $X$, then the resulting detectable region is given by the Ingster-Donoho-Jin expression \prettyref{eq:betaIDJ} scaled by the $L_\infty$-norm of $X$. On the other hand, it is also possible that certain distributions of $\mu_n$ induces different shapes of detectable region than \prettyref{fig:idj}. See Sections \ref{sec:char} and \ref{sec:ex.dilate} for further discussions.

\subsection{Organization}

The rest of the paper is organized as follows.
\prettyref{sec:prob} states the setup, defines the fundamental limit of sparse mixture detection and reviews some previously known results.
The main results of the paper are presented in Sections \ref{sec:main} and \ref{sec:HC}, where we provide an explicit characterization of the optimal detection boundary under mild technical conditions. Moreover, it is shown in \prettyref{sec:HC} that the higher criticism test achieves the optimal performance adaptively. 
\prettyref{sec:ex} particularizes the general result to various special cases to give explicit formulae of the fundamental limits. 
Discussions of generalizations and open problems are presented in \prettyref{sec:discuss}.
The main theorems are proven in \prettyref{sec:pf}, while the proofs of the technical lemmas are relegated to the appendices.

\subsection{Notations}
	\label{sec:notation}

Throughout the paper, $\Phi$ and $\varphi$ denote the cumulative distribution function (CDF) and the density of the standard normal distribution respectively. Let $\bar{\Phi}=1 -\Phi$. 
Let $P^n$ denote the $n$-fold product measure of $P$.
We say $P$ is absolutely continuous with respect to $Q$, denoted by $P \ll Q$, if $P(A)=0$ for any measurable set $A$ such that $Q(A) = 0$.  
We say $P$ is singular with respect to $Q$, denoted by $P \perp Q$, if there exists a measurable $A$ such that $P(A)=1$ and $Q(A) = 0$.
We denote $a_n = o(b_n)$ if $\limsup_{n\diverge} \frac{|a_n|}{|b_n|} = 0$, $a_n = \omega(b_n)$ if $b_n = o(a_n)$, $a_n = O(b_n)$ if $\limsup_{n\diverge} \frac{|a_n|}{|b_n|} < \infty$ and $a_n = \Omega(b_n)$ if $b_n = O(a_n)$.
These asymptotic notations extend naturally to probabilistic setups, denoted by $o_{\Prob}, \omega_{\Prob}$, etc., where limits are in the sense of convergence in probability. 

\section{Fundamental limits and characterization}
	\label{sec:prob}
	
In this section we define the fundamental limits for testing the hypotheses \prettyref{eq:HT.nG} in terms of the sparsity parameter $\beta$. An equivalent characterization in terms of the Hellinger distance is also given.
	
%
%

\subsection{Fundamental limits of detection}
\label{sec:FL}
It is easy to see that as the non-null proportion $\epsilon_n$ decreases, the signal is more sparse and the testing problem in \prettyref{eq:HT.nG} becomes more difficult. Recall that $\epsilon_n$ is given by \prettyref{eq:epsn}
where $\beta \ge 0$ parametrizes the sparsity level.
Thus, the question of detectability boils down to characterizing the smallest (resp. largest) $\beta$ such that the hypotheses in \prettyref{eq:HT.nG} can be distinguished with probability tending to one (resp. zero), when the sample size $n$ is large. 

For testing between two probability measures $P$ and $Q$, denote the optimal sum of Type-I and Type-II error probabilities by
\begin{equation}
	\calE(P,Q) \triangleq \inf_{A} \{P(A) + Q(\comp{A})\},
	\label{eq:perr}
\end{equation}
where the infimum is over all measurable sets $A$. By the Neyman-Pearson Lemma \cite{NP33}, $\calE(P,Q)$ is achieved by the likelihood ratio test: 
declare $P$ if and only if $\fracd{P}{Q} \geq 1$.
Moreover, $\calE(P,Q)$ can be expressed in terms of the \emph{total variation distance} 
\begin{equation}
\TV(P,Q) \triangleq 	\sup_{A} |P(A) - Q(A)| = \frac{1}{2} \int |\diff P - \diff Q|
	\label{eq:TV}
\end{equation}
as
\begin{equation}
\calE(P,Q) = 1 - \TV(P,Q).	
	\label{eq:ETV}
\end{equation}

For a fixed sequence $\{(Q_n,G_n)\}$, denote the total variation between the null and alternative by
\begin{equation}
	V_n(\beta) \triangleq \TV(Q_n^n,\;  ((1-n^{-\beta})Q_n + n^{-\beta} G_n)^n),
	\label{eq:Vn}
\end{equation}
which takes values in the unit interval. 
In view of \prettyref{eq:ETV}, the fundamental limits of testing the hypothesis \prettyref{eq:HT.nG} are defined as follows.
\begin{definition}
\begin{align}
\betal \triangleq & ~\sup\sth{ \beta \geq 0: V_n(\beta) \to 1 }, 	\label{eq:betal}\\
\betau \triangleq & ~\inf\sth{ \beta \geq 0: V_n(\beta) \to 0} .	\label{eq:betau}
\end{align}
If $\overline{\beta}^* = \underline{\beta}^*$, the common value is denoted by $\beta^*$.
	\label{def:betas}
\end{definition}

As illustrated by \prettyref{fig:beta}, the operational meaning of $\betal$ and $\betau$ are as follows: 
for any $\beta > \betau$, all sequences of tests have vanishing probability of success; for any $\beta < \betal$, there exists a sequence of tests with vanishing probability of error. In information-theoretic parlance, if $\overline{\beta}^* = \underline{\beta}^* = \beta^*$, we say \emph{strong converse} holds, in the sense that if $\beta > \beta^*$, all tests fail with probability tending to one; if $\beta < \beta^*$, there exists a sequence of tests with vanishing error probability.

\begin{figure}[ht]
	\centering
%
	\includegraphics[width=.6\columnwidth]{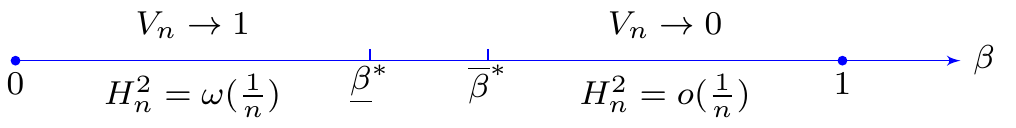}
	\caption{Critical values of $\beta$ and regimes of (in)distinguishability of the hypotheses \prettyref{eq:HT} in the large-$n$ limit.}
	\label{fig:beta}
\end{figure}

Clearly, $\betau$ and $\betal$ only depend on the sequence $\{(Q_n, G_n)\}$.
The following lemma, proved in \prettyref{app:H2mix}, shows that it is always sufficient to restrict the range of $\beta$ to the unit interval.
\begin{lemma}
\begin{equation}
	0 \leq \underline{\beta}^* \leq \overline{\beta}^* \leq 1.
	\label{eq:beta01}
\end{equation}	
	\label{lmm:beta01}
\end{lemma}

In the Gaussian mixture model with $Q_n = \calN(0,1)$, if the sequence $\{G_n\}$ is parametrized by some parameter $r$, the fundamental limit $\beta^*$ in \prettyref{def:betas} is a function of $r$, denoted by $\beta^*(r)$. 
For example, in the Ingster-Donoho-Jin setup \prettyref{eq:HT.IDJ} where $G_n = \calN(\mu_n, 1)$, 
$\beta^*$, denoted by $\beta^*_{\rm IDJ}$, can be obtained by inverting \prettyref{eq:rstar}:
\begin{equation}
	\beta^*_{\rm IDJ}(r)
	= \begin{cases}
	\frac{1}{2}+r & 0 < r \leq \frac{1}{4} \\
	1-(1-\sqrt{r})_+^2 & r > \frac{1}{4} 
\end{cases}.
	\label{eq:betaIDJ}
\end{equation}
In terms of \prettyref{eq:betaIDJ}, the detectable region \prettyref{eq:detregion} is given by the strict hypograph $\{(r,\beta): \beta < \beta^*(r)\}$.
The function $\beta^*_{\rm IDJ}$, plotted in \prettyref{fig:idj}, plays an important role in our later derivations. 
Similarly, for the heteroscedastic mixture \prettyref{eq:HT.CJJ}, inverting \prettyref{eq:cjjr} gives
\begin{align}
\beta^*(r,\sigma^2)
= & ~ \begin{cases}
	\frac{1}{2}+\frac{r}{2-\sigma^2} & 2 \sqrt{r} + \sigma^2 \leq 2 \\
	1-\frac{(1-\sqrt{r})_+^2}{\sigma^2} & 2 \sqrt{r} + \sigma^2 > 2
\end{cases}.	
\end{align}
As shown in \prettyref{sec:ex}, all the above results can be obtained in a unified manner as a consequence of the general results in \prettyref{sec:main}.


\begin{figure}[ht]
	\centering
	\includegraphics[width=.5\columnwidth]{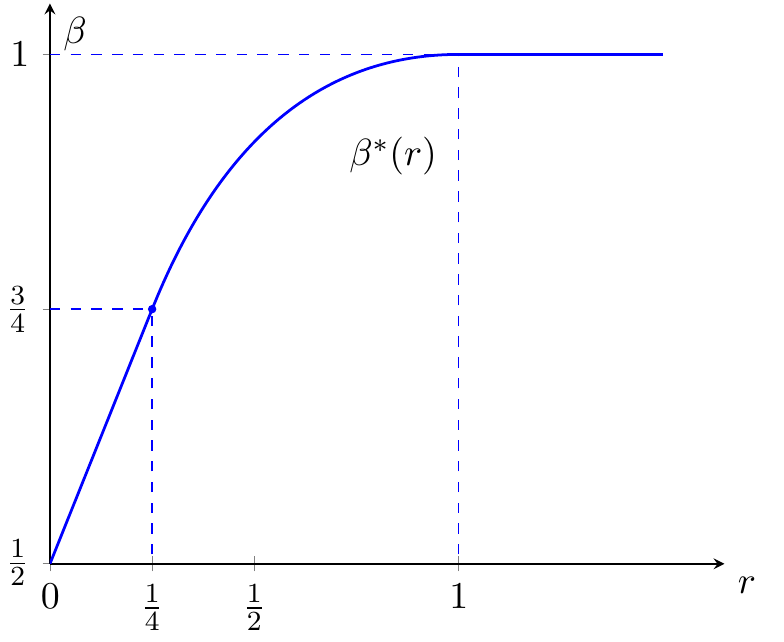}
	\caption{Ingster-Donoho-Jin detection boundary \prettyref{eq:betaIDJ} and the detectable region (below the curve).}
	\label{fig:idj}
\end{figure}

\subsection{Equivalent characterization via the Hellinger distance}
	\label{sec:hel}

Closely related to the total variation distance is the Hellinger distance \cite[Chapter 2]{Lecam86}
 \[
 H^2(P,Q) \triangleq \int (\sqrt{\diff P} - \sqrt{\diff Q})^2,
 \]
which takes values in the interval $[0,2]$ and satisfies the following relationship:
\begin{equation}
	\frac{1}{2} H^2(P,Q) \leq \TV(P,Q) \leq H(P,Q)\sqrt{1 - \frac{H^2(P,Q)}{4}} \leq 1.
	\label{eq:tvh}
\end{equation}
Therefore, the total variation distance converges to zero (resp. one) is equivalent to the squared Hellinger distance converges to zero (resp. two). 
We will be focusing on the Hellinger distance partly due to the fact that it tensorizes nicely under the product measures:
\begin{equation}
H^2(P^n,Q^n) = 2 - 2\left(1 - \frac{H^2(P,Q)}{2}\right)^n.
	\label{eq:hprod}
\end{equation}

Denote the Hellinger distance between the null and the alternative by
\begin{equation}
H_n^2(\beta) \triangleq	H^2(Q_n, (1-n^{-\beta}) Q_n + n^{-\beta} G_n).
	\label{eq:Hnbeta}
\end{equation}
In view of \prettyref{eq:betal} -- \prettyref{eq:betau} and \prettyref{eq:hprod}, the fundamental limits $\betau$ and $\betal$ can be equivalently defined as follows in terms of the asymptotic squared Hellinger distance:
\begin{align}
\betal = & ~\sup\sth{ \beta \geq 0: H_n^2(\beta) = \omega(n^{-1}) }, 	\label{eq:betalH}\\
\betau = & ~\inf\sth{ \beta \geq 0: H_n^2(\beta) = o(n^{-1}) } .	\label{eq:betauH}
\end{align}

\section{Main results}
\label{sec:main}

In this section we characterize the detectable region explicitly by analyzing the exact asymptotics of the Hellinger distance induced by the sequence of distributions $\{(Q_n,G_n)\}$.

\subsection{Characterization of $\beta^*$ for Gaussian mixtures}
	\label{sec:char}
	This subsection we focus on the case of sparse normal mixture with $Q_n = \calN(0, 1)$ and $G_n$  absolutely continuous.
	 We will argue in \prettyref{sec:decom} that by performing the Lebesgue decomposition on $G_n$ if necessary, we can reduce the general problem to the absolutely continuous case.
	 	
	We first note that the \emph{essential supremum} of a measurable function $f$ with respect to a measure $\mu$ is defined as
	\[
	\esssup_{x} f(x) \triangleq \inf\{a\in \reals: \mu(\{x: f(x)>a\})=0\}.
	\]
	We omit mentioning $\mu$ if $\mu$ is the Lebesgue measure.
	Now we are ready to state the main result of this section.


\begin{theorem}
Let $Q_n = \calN(0,1)$. 
Assume that $G_n$ has a density $g_n$ with respect to the Lebesgue measure. Denote the log-likelihood ratio by
\begin{equation}
\ell_n	 \triangleq \log \frac{g_n}{\varphi}.
	\label{eq:ln}
\end{equation}
Let $\alpha: \reals \to \reals$ be a measurable function and define
\begin{equation}
	\beta^{\sharp}
	= \frac{1}{2} + 0 \vee \esssup_{u \in \reals} \sth{\alpha(u)- u^2 + \frac{u^2 \wedge 1}{2}}.	
	\label{eq:main}
\end{equation}
\begin{enumerate}
	\item If
\begin{equation}
	\begin{aligned}
\liminf_{n \diverge} \frac{\ell_{n}(u \sqrt{2 \log n})}{\log n} \geq & ~ \alpha(u)
	\label{eq:upn}
\end{aligned}
\end{equation}
\emph{uniformly} in $u \in \reals$, where $\alpha > 0$ on a set of positive Lebesgue measure, then $\betal \geq \beta^{\sharp}$.

\item If
\begin{align}
\limsup_{n \diverge} \frac{\ell_{n}(u \sqrt{2 \log n})}{\log n} \leq & ~ \alpha(u)
\label{eq:lown}
\end{align}
\emph{uniformly} in $u \in \reals$, then $\betau \leq \beta^{\sharp}$. 
\end{enumerate}
Consequently, if the limits in \prettyref{eq:upn} and \prettyref{eq:lown} agree and $\alpha > 0$ on a set of positive measure, then $\beta^* = \beta^{\sharp}$.
	\label{thm:main}
\end{theorem}
\begin{proof}
	\prettyref{sec:pfmain}.
\end{proof}

	Assuming the setup of \prettyref{thm:main}, we ask the following question in the reverse direction: What kind of function $\alpha$ can arise in equations \prettyref{eq:upn} and \prettyref{eq:lown}? The following lemma (proved in \prettyref{sec:pfmain}) gives a necessary and sufficient condition for $\alpha$. 	However, in the special case of convolutional models, the function $\alpha$ needs to satisfy more stringent conditions, which we also discuss below.
	\begin{lemma}
	Suppose
	\begin{align}
\lim_{n \diverge} \frac{\ell_n(u \sqrt{2 \log n })}{\log n} = & ~ \alpha(u),	\label{eq:alphau}
\end{align}
holds \emph{uniformly} in $u \in \reals$ for some measurable function $\alpha: \reals \to \reals$. Then
\begin{equation}
\lim_{t \diverge} \frac{1}{t} \log \int_\reals \exp(t (\alpha(u) - u^2)) \diff u = 0.
	\label{eq:alpha.legit}
\end{equation}
In particular, $\alpha(u) \leq u^2$ Lebesgue-a.e. Conversely, for all measurable $\alpha$ that satisfies \prettyref{eq:alpha.legit}, there exists a sequence of $\{G_n\}$, such that \prettyref{eq:alphau} holds.

Additionally, if the model is convolutional, \ie, $G_n = P_n * \calN(0,1)$, then $\alpha$ is convex.

	\label{lmm:alpha}
\end{lemma}


 In many applications, we want to know how fast the optimal error probability decays if $\beta$ lies in the detectable region. The following result gives the precise asymptotics for the Hellinger distance, which also gives upper bounds on the total variation, in view of \prettyref{eq:tvh}.
 \begin{theorem}
Assume that \prettyref{eq:alphau} holds. For any $\beta \geq \frac{1}{2}$, the exponent of the Hellinger distance \prettyref{eq:Hnbeta} is given by
\begin{equation}
	\lim_{n \diverge}\frac{\log H^2_n(\beta)}{\log n} = \sfE(\beta),
	\label{eq:Hnexp}
\end{equation}
where
\begin{align}
\sfE(\beta)
=	& ~ \esssup_{u \in \reals} \{(2(\alpha(u) -\beta))\wedge (\alpha(u) -\beta) - u^2 \} 	\label{eq:Ebeta}\\
=	& ~ \esssup_{u: \alpha(u) \leq \beta} \{2\alpha(u) -2 \beta - u^2 \} \vee \esssup_{u: \alpha(u) > \beta} \{\alpha(u) -\beta - u^2 \} 	\label{eq:Ebeta2}
\end{align}
which satisfies $\sfE(\beta) > -1$ (resp. $\sfE(\beta) < -1$)  if and only if $\beta <\beta^\sharp$ (resp. $\beta >\beta^\sharp$).
	\label{thm:Ebeta}
\end{theorem}
 
 As an application of \prettyref{thm:main}, the following result relates the fundamental limit $\beta^*$ of the convolutional models to the classical Ingster-Donoho-Jin detection boundary:
\begin{coro}
Let $G_n = P_n * \calN(0,1)$. Assume that $P_n$	has a density $p_n$ which satisfies that
\begin{equation}
\lim_{n \diverge} \frac{\log p_n(t \sqrt{2\log n})}{\log n} = -f(t)	
	\label{eq:ft}
\end{equation}
uniformly in $t \in \reals$ for some measurable $f: \reals \to \reals$. Then
\begin{align}
\beta^* = \esssup_{t \in \reals} \{\beta^*_{\rm IDJ}(t^2) - f(t)\} \label{eq:betaconv}
\end{align}
where $\beta^*_{\rm IDJ}$ is the Ingster-Donoho-Jin detection boundary defined in \prettyref{eq:betaIDJ}.
	\label{cor:conv}
\end{coro}

It should be noted that the convolutional case of the normal mixture detection problem is briefly discussed in \cite[Section 6.1]{CJJ11}, where inner and outer bounds on the detection boundary are given but do not meet. Here \prettyref{cor:conv} completely settles this question. 
See \prettyref{sec:ex} for more examples.

We conclude this subsection with a few remarks on \prettyref{thm:main}.
\begin{remark}[Extremal cases]
Under the assumption that the function $\alpha > 0$ on a set of positive Lebesgue measure, 
the formula \prettyref{eq:main} shows that the fundamental limit $\beta^*$ lies in the very sparse regime ($\frac{1}{2} \leq \beta^* \leq 1$). 
We discuss the two extremal cases as follows: 
\begin{enumerate}
	\item \emph{Weak signal}: Note that $\beta^* = \frac{1}{2}$ if and only if $\alpha(u) \leq u^2 - \frac{u^2 \wedge 1}{2}$ almost everywhere. In this case the non-null effect is too weak to be detected for any $\beta > \frac{1}{2}$. One example is the zero-mean heteroscedastic case $G_n = \calN(0,\sigma^2)$ with $\sigma^2 \leq 2$. Then we have $\alpha(u) \leq \frac{u^2}{2}$. 
	
	\item \emph{Strong signal}: Note that $\beta^* = 1$ if and only if there exists $u$, such that $|u| \geq 1$ and
	\begin{equation}
	\alpha(u) = u^2.
	\label{eq:beta1}
\end{equation}
	 At this particular $u$, the density of the signal satisfies $g_n(u \sqrt{2 \log n})=n^{-o(1)}$, which implies that there exists significant mass beyond $\sqrt{2 \log n}$, the extremal value under the null hypothesis \cite{dHF06}. This suggests the possibility of constructing test procedures based on the \emph{sample maximum}. 
Indeed, to understand the implication of \prettyref{eq:beta1} more quantitatively, let us look at an even weaker condition: there exists $u$ such that $|u| \geq 1$ and
\begin{equation}
\limsup_{n\diverge} \frac{1}{\log n} \log \frac{1}{\prob{u^{-1} Y_n \geq  \sqrt{2 \log n}}} = 0,
	\label{eq:tail}
\end{equation}
which, as shown in \prettyref{app:tail}, implies that $\beta^* = 1$. 

\end{enumerate}

	\label{rmk:extreme}
\end{remark}


\begin{remark}
	In general $\beta^*$ need not exist. Based on \prettyref{thm:main}, it is easy to construct a Gaussian mixture where $\betau$ and $\betal$ do not coincide.
	 For example, let $\alpha_0$ and $\alpha_1$ be two measurable functions which satisfy \prettyref{lmm:alpha} and give rise to different values of $\betas$ in \prettyref{eq:main}, which we denote by $\betas_0 < \betas_1$. Then there exist sequences of distributions $\{G_n^{(0)}\}$ and $\{G_n^{(1)}\}$ which satisfy \prettyref{eq:alphau} for $\alpha_0$ and $\alpha_1$ respectively. Now define $\{G_n\}$ by $G_{2k}=G_k^{(0)}$ and $G_{2k+1}=G_k^{(1)}$. Then by \prettyref{thm:main}, we have $\betal = \betas_0 < \betau = \betas_1$.
\end{remark}




\subsection{Non-Gaussian mixtures}
	\label{sec:nong}
	The detection boundary in \cite{Ingster97,DJ.HC} is obtained by deriving the limiting distribution of the log-likelihood ratio which relies on the normality of the null hypothesis. In contrast, our approach is based on analyzing the sharp asymptotics of the Hellinger distance. This method enables us to generalize the result of \prettyref{thm:main} to sparse non-Gaussian mixtures \prettyref{eq:HT.nG}, where we even allow the null distribution $Q_n$ to vary with the sample size $n$. 

\begin{theorem}
Consider the hypothesis testing problem \prettyref{eq:HT.nG}.
Let $G_n \ll Q_n$. Denote by $F_n$ and $z_n$ the CDF and the quantile function of $G_n$, respectively, \ie,
\begin{equation}
	z_n(p) = \inf\{y \in \reals: F_n(y) \geq p\}.
	\label{eq:zn}
\end{equation}
If the log-likelihood ratio
\begin{equation}
\ell_n	= \log \fracd{G_n}{Q_n}
	\label{eq:ln.nG}
\end{equation}
satisfies
\begin{align}
\lim_{n \diverge} \sup_{s \geq (\log_2 n)^{-1}} \left|\frac{\ell_n(z_n(n^{-s})) \vee \ell_n(z_n(1-n^{-s}))}{\log n} - \gamma(s) \right| = & ~ 0
\label{eq:gammas}
\end{align}
as $n \diverge$
\emph{uniformly} in $s \in \reals_+$ for some measurable function $\gamma: \reals_+ \to \reals$. If $\gamma > 0$ on a set of positive Lebesgue measure, then
\begin{equation}
	\beta^* = \frac{1}{2} + 0 \vee \esssup_{s \geq 0} \sth{\gamma(s) - s + \frac{s \wedge 1}{2}} .
	\label{eq:nG}
\end{equation}	
\label{thm:nG}
\end{theorem}

The function $\gamma$ appearing in \prettyref{thm:nG} satisfies the same condition as in \prettyref{lmm:alpha}. 
Comparing \prettyref{thm:nG} with \prettyref{thm:main}, we see that the uniform convergence condition \prettyref{eq:alphau} is naturally replaced by the uniform convergence of the log-likelihood ratio evaluated at the null quantile. Using the fact that $\frac{z}{1+z^2} \leq \frac{\bar{\Phi}(z)}{\varphi(z)} \leq \frac{1}{z}$ for all $z>0$ \cite[7.1.13]{AS}, which implies that
\begin{equation}
\bar{\Phi}(z) = \frac{\varphi(z)}{z}(1+o(1))
\label{eq:phic}
\end{equation}	
uniformly as $z \diverge$, 
we can recover \prettyref{thm:main} from \prettyref{thm:nG} by setting $\gamma(s) = \alpha(-\sqrt{s}) \vee \alpha(\sqrt{s})$.

\subsection{Decomposition of the alternative}
	\label{sec:decom}
	The results in \prettyref{thm:main} and \prettyref{thm:nG} are obtained under the assumption that the non-null effect $G_n$ is absolutely continuous with respect to the null distribution $Q_n$. Next we show that it does not lose generality to focus our attention on this case. Using the Hahn-Lebesgue decomposition \cite[Theorem 1.6.3]{evans.gariepy}, we can write 
\begin{equation}
G_n = (1-\kappa_n) G_n' + \kappa_n \nu_n	
	\label{eq:Gndecomp}
\end{equation}
 for some $\kappa_n \in [0,1]$, where $G_n' \ll Q_n$ and $\nu_n \perp Q_n$.
 Put 
 \begin{equation}
\epsilon_n' = \frac{\epsilon_n(1-\kappa_n)}{1-\epsilon_n \kappa_n}	\quad\mbox{and}\quad
Q_n'	=  (1-\epsilon_n') Q_n + \epsilon_n' G_n', \label{eq:epsilonp-Gnp}
\end{equation}
which satisfies $Q_n' \ll Q_n$. 
Then 
$
(1-\epsilon_n) Q_n + \epsilon_n G_n
=  (1-\epsilon_n \kappa_n) Q_n' + \epsilon_n \kappa_n \nu_n.
$
By \prettyref{lmm:H2.mix}, 
\begin{equation}
	H^2(Q_n,(1-\epsilon_n) Q_n + \epsilon_n G_n) = \Theta(\epsilon_n \kappa_n \vee H^2((1-\epsilon') Q_n + \epsilon_n' G_n'))
	\label{eq:sand}
\end{equation}
Therefore the asymptotic Hellinger distance of the original problem is completely determined by $\epsilon_n \kappa_n$ and the square-Hellinger distance $H^2((1-\epsilon') Q_n + \epsilon_n' G_n')$, which is also of a \emph{sparse mixture} form, with $(\epsilon_n,G_n)$ replaced by $(\epsilon_n',G_n')$ given in \prettyref{eq:epsilonp-Gnp}.
In particular, we note the following special cases:
\begin{enumerate}
	\item If $\epsilon_n \kappa_n = O(n^{-1})$, then $H^2(Q_n,(1-\epsilon_n) Q_n + \epsilon_n G_n) = o(n^{-1})$ (resp. $\omega(n^{-1})$) if and only if $H^2(Q_n,(1-\epsilon_n') Q_n + \epsilon_n' G_n') = o(n^{-1})$ (resp. $\omega(n^{-1})$),
	which means that detectability of the original sparse mixture coincide with the new mixture.
	\item If $\epsilon_n \kappa_n = \omega(n^{-1})$, then $H^2(Q_n,(1-\epsilon_n) Q_n + \epsilon_n G_n)=\omega(n^{-1})$, which means that the original sparse mixture can be detected reliably. In fact, a trivial optimal test is to reject the null hypothesis if there exists one sample lying in the support of the singular component $\nu_n$.
\end{enumerate}

\section{Adaptive optimality of Higher Criticism tests}
\label{sec:HC}
As discussed in \prettyref{sec:FL}, the fundamental limit $\beta^*$ of testing sparse normal mixtures \prettyref{eq:HT} can be achieved by the likelihood ratio test. However, in general the likelihood ratio test requires the knowledge of the alternative distribution, 
which is typically not accessible in practice. To overcome this limitation, it is desirable to construct \emph{adaptive} testing procedures to achieve the optimal performance simultaneously for a collection of alternatives. This problem is also known as \emph{universal hypothesis testing}. See, \eg, \cite{Hoeffding65,ZZM92,UHMV11} and the references therein, for results on discrete alphabets. The basic idea of adaptive procedures usually involves comparing the empirical distribution of the data to the null distribution, which is assumed to be known.

For the problem of detecting sparse normal mixtures, it is especially relevant to construct adaptive procedures, since in practice the underlying sparsity level and the non-zero priors are usually unknown. Toward this end, Donoho and Jin \cite{DJ.HC} introduced an adaptive test based on Tukey's \emph{higher criticism} statistic. For the special case of \prettyref{eq:HT.IDJ}, \ie, $P_n = \delta_{\sqrt{2 r \log n}}$, it is shown that the higher criticism test achieves the optimal detection boundary \prettyref{eq:betaIDJ} while being adaptive to the unknown non-null parameters  $(\beta,r)$. Following the generalization by Jager and Wellner \cite{JW07} via \renyi divergence, next we explain briefly the gist of the higher criticism test.

Given the data $\ntok{Y_1}{Y_n}$, denote 
the empirical CDF by
\[
\bbF_n(t) = \frac{1}{n}\sum_{i=1}^n \indc{Y_i \leq t},
\]
respectively. Similar to the Kolmogorov-Smirnov statistic \cite[p. 91]{Shorack.Wellner} which computes the $L_\infty$-distance (maximal absolute difference) between the empirical CDF and the null CDF, the higher criticism statistic is the maximal pointwise $\chi^2$-divergence between the null and the empirical CDF. We first introduce a few auxiliary notations. Recall that the $\chi^2$-divergence between two probability measures is defined as
\[
\chi^2(P \, || \, Q) \triangleq \int \pth{\fracd{P}{Q} - 1}^2 \diff Q.
\]
In particular, the binary $\chi^2$-divergence function (\ie, the $\chi^2$-divergence between Bernoulli distributions) is given by
\[
\chi^2(\Bern(p) \, || \, \Bern(q)) = \frac{(p-q)^2}{q(1-q)},
\]
where $\Bern(p)$ denotes the Bernoulli distribution with bias $p$. The higher criticism statistic is defined by
\begin{align}
\HC_n 
\triangleq & ~ 	\sup_{t \in \reals} \sqrt{n \chi^2(\Bern(\bbF_n(t)) \, || \, \Bern(\Phi(t)))} \label{eq:HC}\\
= & ~ \sqrt{n} 	\sup_{t \in \reals} \frac{|\bbF_n(t) - \Phi(t)|}{\sqrt{\Phi(t) \bar{\Phi}(t)}} \label{eq:HCsupt}
\end{align}
Based on the statistics \prettyref{eq:HC}, the higher criticism test declares $H_1$ if and only if
\begin{equation}
\HC_n > \sqrt{2 (1+\delta) \log \log n}	
	\label{eq:HCtest}
\end{equation}
where $\delta > 0$ is an arbitrary fixed constant.

The next result shows that the higher criticism test achieves the fundamental limit $\beta^*$ characterized by \prettyref{thm:main} while being adaptive to all sequences of distributions $\{G_n\}$ which satisfy the regularity condition \prettyref{eq:alphau}. This result generalizes the adaptivity of the higher criticism procedure far beyond the original equal-signal-strength setup in \cite{DJ.HC} and the heteroscedastic extension in \cite{CJJ11}.

\begin{theorem}
	Under the same assumption of \prettyref{thm:main}, for any $\beta > \beta^*$, the sum of Type-I and Type-II error of the higher criticism test \prettyref{eq:HCtest} vanishes as $n\diverge$.
	\label{thm:HC}
\end{theorem}


\section{Examples}
	\label{sec:ex}
In this section we particularize the general result in \prettyref{thm:main} to several interesting special cases to obtain explicit detection boundaries. 
\subsection{Ingster-Donoho-Jin detection boundary}
\label{sec:ex.IDJ}
We derive the classical detection boundary \prettyref{eq:betaIDJ} from \prettyref{thm:main} for the equal-signal-strength setup \prettyref{eq:HT.IDJ}, which is a convolutional model with signal distribution
\begin{equation}
P_n = \delta_{\mu_n}	
	\label{eq:IDJPn}
\end{equation}
 and $\mu_n$ in \prettyref{eq:r}.
The log-likelihood ratio is given by  
\[
\ell_n(y) = \log \frac{\varphi(y-\mu_n)}{\varphi(y)} =  - \frac{\mu_n^2}{2} + \mu_n y = -r \log n + \sqrt{2 r \log n} \, y.
\]
 Plugging in $y = u \sqrt{2 \log n}$, we have $\ell_n(u  \sqrt{2 \log n}) = - r \log n  + 2u \sqrt{r} \log n$. Consequently, the condition \prettyref{eq:alphau} is fulfilled uniformly in $u \in \reals$ with
\begin{equation}
\alpha(u) = 2 u \sqrt{r}-r.
	\label{eq:alphaidj}
\end{equation}
Straightforward calculation yields that
\begin{equation}
\esssup_{u \in \reals} \sth{2u\sqrt{r}  - r - u^2 + \frac{u^2 \wedge 1}{2}} = 
\begin{cases}
	r & 0 < r \leq \frac{1}{4} \\
	\frac{1}{2}-(1-\sqrt{r})_+^2 & r > \frac{1}{4}.
\end{cases}	
	\label{eq:betaIDJ1}
\end{equation}
Applying \prettyref{thm:main}, we obtain the desired expression \prettyref{eq:betaIDJ} for $\beta^*(r)$.

As a variation of \prettyref{eq:IDJPn}, the symmetrized version of \prettyref{eq:IDJPn}
\begin{equation}
P_n = \frac{1}{2}(\delta_{\mu_n} + \delta_{-\mu_n})	
	\label{eq:rademacher}
\end{equation}
was considered in \cite[Section 8.1.6]{IS03}, whose detection boundary is shown to be identical to \prettyref{eq:betaIDJ}.
Indeed, for binary-valued signal distributed according to \prettyref{eq:rademacher},
we have
\begin{align}
\ell_n(u \sqrt{2 \log n})
= & ~ -\frac{\mu_n^2}{2} + \log \cosh (\mu_n u \sqrt{2 \log n}) \nonumber \\
= & ~ -r \log n + \log (n^{2u\sqrt{r}} + n^{-2u\sqrt{r}}) -\log 2 \nonumber
\end{align}
which gives rise to
\begin{equation}
\alpha(u) = 2 |u| \sqrt{r}-r	
	\label{eq:alphaIS}
\end{equation}
Comparing \prettyref{eq:alphaIS} with \prettyref{eq:alphaidj} and \prettyref{eq:betaIDJ1}, we conclude that 
the detection boundary \prettyref{eq:betaIDJ} still applies. 

\subsection{Dilated signal distributions}
	\label{sec:ex.dilate}
Generalizing both the unary and binary signal distributions in \prettyref{sec:ex.IDJ}, we consider $P_n$ that is the distribution of the random variable 
\begin{equation}
	X_n = \mu_n X
	\label{eq:dilate}
\end{equation}
where $\mu_n > 0$ is a sequence of positive numbers and $X$ is distributed according to a fixed distribution $P$, parameterizing the shape of the signal. In other words, $P_n$ is the dilation of $P$ by $\mu_n$. We ask the following question: By choosing the sequence $\mu_n$ and the random variable $X$, is it possible to have detection boundaries which are shaped differently than the classical Ingster-Donoho-Jin detection boundary?

It turns out that for $\mu_n = \sqrt{2 \log n}$,
the answer to the above question is negative. As the next theorem shows, the detection boundary is given by that of the classical setup rescaled by the $L_\infty$-norm of $X$. Note that \prettyref{eq:IDJPn} and \prettyref{eq:rademacher} corresponds to $P = \delta_{\sqrt{r}}$ and $P = \frac{1}{2} (\delta_{\sqrt{r}} + \delta_{-\sqrt{r}})$, respectively.
\begin{coro}
Consider the convolutional model $G_n = P_n * \calN(0,1)$, where $P_n$ is the distribution of $\sqrt{2 \log n} X$. Then
\begin{equation}
	\beta^* = \beta^*_{\rm IDJ}(\linf{X}^2) =
	\begin{cases}
	\linf{X}^2 + \frac{1}{2} & 0 < \linf{X} \leq \frac{1}{2} \\
	1-(1-\linf{X})_+^2 & \linf{X} > \frac{1}{2}.
\end{cases}
	\label{eq:betadilate}
\end{equation}
	\label{cor:dilate}
\end{coro}
\begin{proof}
Recall that $\beta^*_{\rm IDJ}(\cdot)$ denotes the Ingster-Donoho-Jin detection boundary defined in \prettyref{eq:betaIDJ}.
Since the log-likelihood ratio is given by $\ell_n(y) = \expect{\exp(-\frac{X_n^2}{2} + X_n y)}$, we have
\begin{align}
\ell_n(u \sqrt{2\log n})
= & ~ \log \expect{n^{-X^2 + 2 u X}}	\nonumber \\
= & ~ \esssup_X \sth{-X^2 + 2 u X} \log n (1+o(1)) \label{eq:alphadilate}	,
\end{align}
where we have applied \prettyref{lmm:laplace} and the essential supremum in \prettyref{eq:alphadilate} is with respect to $P$, the distribution of $X$. Therefore $\alpha(u) = \esssup_X \sth{-X^2 + 2 u X}$. Applying \prettyref{thm:main} yields the existence of $\beta^*$, given by
\begin{align}
\beta^*
= & ~ \frac{1}{2} + \esssup_{u \in \reals} \sth{\esssup_X \sth{-X^2 + 2 u X} - u^2 + \frac{u^2 \wedge 1}{2}} \nonumber \\
= & ~ \frac{1}{2} + \esssup_X \esssup_{u \in \reals} \sth{-X^2 + 2 u X - u^2 + \frac{u^2 \wedge 1}{2}} \nonumber \\
= & ~ \esssup_X \beta^*_{\rm IDJ}(X^2) \nonumber \\
= & ~ \beta^*_{\rm IDJ}(\linf{X}^2), \label{eq:betadilate0}
\end{align}
where \prettyref{eq:betadilate0} follows from the facts that $\beta^*_{\rm IDJ}(\cdot)$ is increasing and that $\linf{X} = \esssup |X|$. 
\end{proof}

\begin{remark}
\prettyref{cor:dilate} tightens the bounds given at the end of \cite[Section 6.1]{CJJ11} based on the interval containing the signal support.
From \prettyref{eq:betadilate} we see that the detection boundary coincides with the classical case with $\sqrt{r}$ replaced by $L_\infty$-norm of $X$. Therefore, as far as the detection boundary is concerned, only the support of $X$ matters and the detection problem is driven by the maximal signal strength. In particular, for $\linf{X} \geq 1$ or non-compactly supported $X$, we obtain the degenerate case $\beta^* = 1$ (see also \prettyref{rmk:extreme} about the strong-signal regime). However, it is possible that the density of $X$ plays a role in finer asymptotics of the testing problem, \eg, the convergence rate of the error probability and the limiting distribution of the log-likelihood ratio at the detection boundary. 
\end{remark}

One of the consequences of \prettyref{cor:dilate} is the following: as long as $\mu_n = \sqrt{2 \log n}$, non-compactly supported $X$ results in the degenerate case of $\beta^* =1$, since the signal is too strong to go undetected. However, this conclusion need not be true if $\mu_n$ behaves differently. We conclude this subsection by constructing a family of distributions of $X$ with unbounded support and an appropriately chosen sequence $\{\mu_n\}$, such that the detection boundary is non-degenerate: Let $X$ be distributed according to the following \emph{generalized Gaussian} (Subbotin) distribution $P_{\tau}$ \cite{Taguchi78} with shape parameter $\tau > 0$, whose density is
\begin{equation}
	p_{\tau}(x) = \frac{\tau}{2 \Gamma(\tau)} \exp(-|x|^{\tau}).
	\label{eq:ggaussian}
\end{equation}
Put $\mu_n = \sqrt{2r} (\log n)^{\frac{1}{2} - \frac{1}{\tau}}$.
Then the density of $X_n$ is given by $v_n(x) = \frac{1}{\mu_n} p(\frac{x}{\mu_n})$. Hence
\[
v_n(t \sqrt{2 \log n}) = \frac{\tau}{2 \Gamma(\tau) \mu_n} n^{- |t|^{\tau} r^{-\frac{\tau}{2}}} , 
\]
which satisfies the condition \prettyref{eq:ft} with
$
f(t) = |t|^{\tau} r^{-\frac{\tau}{2}}. 
$
Applying \prettyref{cor:conv}, we obtain the detection boundary $\beta^*$ (a two-dimensional \emph{surface} parametrized by $(r,\tau)$ shown in \prettyref{fig:ggaussian}) as follows
\begin{equation}
\beta^*
= \sup_{t \in \reals} \{\beta^*_{\rm IDJ}(t^2) - |t|^{\tau} r^{-\frac{\tau}{2}} \} 
= \sup_{z \geq 0} \{\beta^*_{\rm IDJ}(r z^2) - z^{\tau} \} \label{eq:betaggaussian}
\end{equation}
where \prettyref{eq:betaIDJ} is the  Ingster-Donoho-Jin detection boundary.

\medskip

\begin{figure}[htp]
	\centering
	\begin{overpic}[
	width=.6\columnwidth]{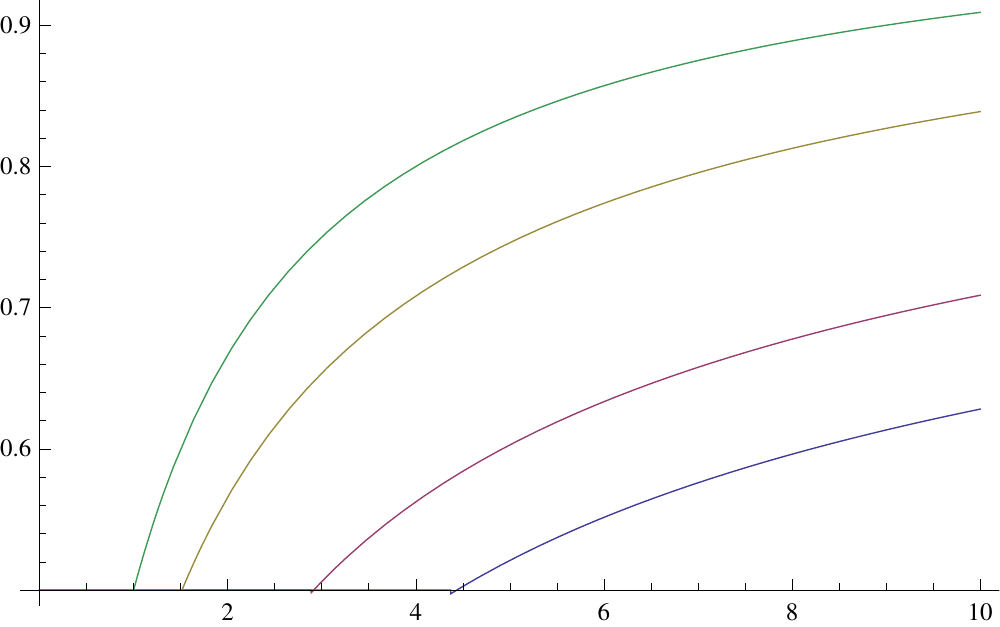}
\put(8,60){$\beta^*$}
\put(102,3){$r$}
\put(60,57){$\tau = 2$}
\put(60,46){$\tau = 1.5$}
\put(60,26){$\tau = 1$}
\put(60,15){$\tau = 0.8$}
\end{overpic}
\caption{Detection boundary $\beta^*$ given by \prettyref{eq:betaggaussian} as a function of $r$ for various values of $\tau$.}
	\label{fig:ggaussian}
\end{figure}
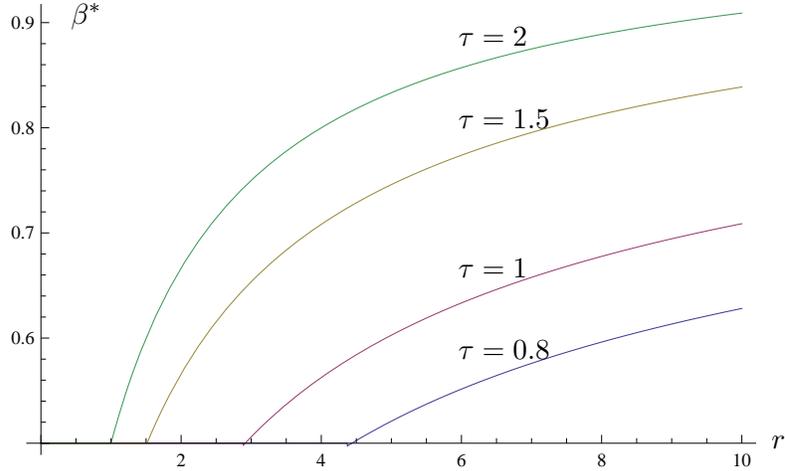

Equation \prettyref{eq:betaggaussian} can be further simplified for the following special cases.
\begin{itemize}
	\item $\tau = 1$ (Laplace):  
	Plugging \prettyref{eq:betaIDJ} into \prettyref{eq:betaggaussian}, straightforward computation yields
	\[
	\beta^*  = \frac{1}{2} \vee \pth{1-\frac{1}{2\sqrt{r}}}_+^2  = 
	\begin{cases}
	\pth{1-\frac{1}{2\sqrt{r}}}^2 & r > \frac{3}{2}+\sqrt{2}	\\
\frac{1}{2} & r \leq \frac{3}{2}+\sqrt{2}
\end{cases}.
	\]
	\item $\tau = 2$ (Gaussian): In this case we have $X \sim \calN(0,\frac{1}{2})$ and $X_n \sim \calN(0,r)$. This is a special case of the heteroscedastic case in \cite{CJJ11}, which will be discussed in detail in \prettyref{sec:ex.hetero}. Simplifying \prettyref{eq:betaggaussian} we obtain
	\[
	\beta^*  = \frac{1}{2} \vee \frac{r}{1+r},
	\]
	which coincides with \prettyref{eq:betacjjr0}.
\end{itemize}

\subsection{Heteroscedastic normal mixture}
\label{sec:ex.hetero}
The heteroscedastic normal mixtures considered in \prettyref{eq:HT.CJJ} corresponds to
\[
G_n = \calN(\mu_n, \sigma^2)	
\]
 with $\mu_n$ given in \prettyref{eq:r} and $\sigma^2 \geq 0$. In particular, if $\sigma^2 \geq 1$, $G_n$ is given by the convolution $G_n = \Phi * P_n$, where the Gaussian component $P_n = \calN(\mu_n, \sigma^2-1)$ models the variation in the signal amplitude. 

For any $u \in \reals$,
\begin{equation}
	\ell_n(u \sqrt{2 \log n}) = \log \frac{\varphi\pth{\frac{u \sqrt{2 \log n}-\mu_n}{\sigma}}}{\varphi(u \sqrt{2 \log n})} 
	= \alpha(u) \log n,
\end{equation}
where
\[
\alpha(u) = u^2 - \frac{(u - \sqrt{r})^2}{\sigma^2}.
\]
Similar to the calculation in \prettyref{sec:ex.IDJ}, we have\footnote{In the first case of \prettyref{eq:sups1} it is understood that $\frac{0}{0}=0$.}
\begin{equation}
\sup_{0 \leq s \leq 1}\sth{\alpha(s)- \frac{s}{2}} = 
\begin{cases}
	\frac{r}{2-\sigma^2} & 2 \sqrt{r} + \sigma^2 \leq 2 \\
	\frac{1}{2}-\frac{(1-\sqrt{r})_+^2}{\sigma^2} & 2 \sqrt{r} + \sigma^2 > 2
\end{cases}	
	\label{eq:sups1}
\end{equation}
and
\begin{equation}
\sup_{s \geq 1}\sth{\alpha(s)-s} = -\frac{(1-\sqrt{r})_+^2}{\sigma^2}.
	\label{eq:sups2}
\end{equation}
Note that $\frac{r}{2-\sigma^2} - (\frac{1}{2}-\frac{(1-\sqrt{r})_+^2}{\sigma^2}) \geq \frac{(\sigma^2+2\sqrt{2}-2)^2}{2\sigma^2(2-\sigma^2)} \geq 0$ if $2 \sqrt{r} + \sigma^2 \leq 2$. Assembling \prettyref{eq:sups1} -- \prettyref{eq:sups2} and applying \prettyref{thm:main}, we have
\begin{align}
\beta^*(r,\sigma^2)
= & ~ \frac{1}{2}+ \pth{\frac{r}{2-\sigma^2}} \vee \pth{\frac{1}{2}-\frac{(1-\sqrt{r})_+^2}{\sigma^2}} \\
= & ~ \begin{cases}
	\frac{1}{2}+\frac{r}{2-\sigma^2} & 2 \sqrt{r} + \sigma^2 \leq 2 \\
	1-\frac{(1-\sqrt{r})_+^2}{\sigma^2} & 2 \sqrt{r} + \sigma^2 > 2.
\end{cases}	
	\label{eq:bscjj}
\end{align}
Solving the equation $\beta^*(r,\sigma^2)=\beta$ in $r$ yields the equivalent detection boundary \prettyref{eq:cjjr} in terms of $r$.
In the special case of $r=0$, where the signal is distributed according to $P_n = \calN(0,\tau^2)$, we have 
\begin{equation}
\beta^*(0,1+\tau^2) = \frac{\tau^2 \vee 1}{1+\tau^2 \vee 1}. 
	\label{eq:betacjjr0}
\end{equation}
Therefore, as long as the signal variance exceeds that of the noise, reliable detection is possible in the very sparse regime $\beta > \frac{1}{2}$, even if the average signal strength does not tend to infinity.

\subsection{Non-Gaussian mixtures}
	\label{sec:ex.nong}


We consider the detection boundary of the following generalized Gaussian location mixture which was studied in \cite[Section 5.2]{DJ.HC}:
\begin{equation}
H_{0}^{(n)}:  Y_i \, \iiddistr \, P_\tau(\cdot) \quad \text{versus}\quad
H_{1}^{(n)}:  Y_i \, \iiddistr \, (1-\epsilon_n) (1-\epsilon_n) P_{\tau}(\cdot) + \epsilon_n P_{\tau}(\cdot-\mu_n)	
	\label{eq:HT.poi}
\end{equation}
where $P_{\tau}$ is defined in \prettyref{eq:ggaussian}, and $\mu_n = (r \log n)^{\frac{1}{\tau}}$. Since $z(1-n^{-s}) = z(n^{-s}) = (s\log n)^{\frac{1}{\tau}}(1+o(1))$ uniformly in $s$, \prettyref{eq:gammas} is fulfilled with $\gamma(s) = s - |s^{\frac{1}{\tau}} - r^{\frac{1}{\tau}}|$. Applying \prettyref{thm:nG}, we have
\begin{align}
\beta^*(r)
= & ~ \frac{1}{2} + 0 \vee \sup_{s \geq 0} \pth{- |s^{\frac{1}{\tau}} - r^{\frac{1}{\tau}}| + \frac{s\wedge 1}{2}}	
= \frac{1}{2} + 0 \vee \sup_{u \geq 0} \pth{- |u - r^{\frac{1}{\tau}}| + \frac{u^\tau \wedge 1}{2}}	\nonumber \\	
= & ~ 
\begin{cases}
1 & r > 1 \\
\frac{1+r}{2} & \tau \leq 1, r \leq 1,  \\
\frac{1}{2} + \frac{\frac{1}{2}-  2^{\frac{\tau}{1-\tau}} }{(1 - 2^{\frac{1}{1-\tau}})^{\tau}} r & \tau \geq 1, r < (1 - 2^{\frac{1}{1-\tau}})^{\tau} \\ 
1 - (1-r^{\frac{1}{\tau}})^{\tau} & \tau \geq 1, r \geq (1 - 2^{\frac{1}{1-\tau}})^{\tau}
\end{cases}.
\label{eq:beta.ggn}
\end{align}
It is easy to verify that \prettyref{eq:beta.ggn} agrees with the results in \cite[Theorem 5.1]{DJ.HC}. Similarly, the detection boundary for exponential-$\chi^2_2$ mixture in \cite[Theorem 1.7]{DJ.HC} can also be derived from \prettyref{thm:nG}.

%

\section{Discussions}
\label{sec:discuss}

We conclude the paper with a few discussions and open problems.

\subsection{Moderately sparse regime $0 \leq \beta \leq \frac{1}{2}$}
	\label{sec:dense}	
	
	Our main results in \prettyref{sec:main} only concern the \emph{very sparse} regime $\frac{1}{2} < \beta < 1$. This is because under the assumption in \prettyref{thm:main} that $\alpha > 0$ on a set of positive Lebesgue measure, we always have $\beta^* \geq \frac{1}{2}$. 
	One of the major distinctions between the very sparse and moderately sparse regimes is the effect of symmetrization. To illustrate this point, consider the sparse normal mixture model \prettyref{eq:HT}. Given any $G_n$, replacing it by its symmetrized version $\tG_n(\diff x) \triangleq \frac{G_n(\diff x)+G_n(-\diff x)}{2}$ always increases the difficulty of testing. This follows from the inequality $H^2(\tG_n, \Phi) \leq H^2(G_n,\Phi)$, a consequence of the convexity of the squared Hellinger distance and the symmetry of $\Phi$. 
	A natural question is: Does symmetrization always have an impact on the detection boundary? In the very sparse regime, it turns out that under the regularity conditions imposed in \prettyref{thm:main}, symmetrization does not affect the fundamental limit $\beta^*$, because both $G_n$ and $\tilde{G}_n$ give rise to the same function $\alpha$. It is unclear whether $\betau$ and $\betal$ remain unchanged if an arbitrary sequence $\{G_n\}$ is symmetrized. However, in the moderately sparse regime, an asymmetric non-null effect can be much more detectable than its symmetrized version.
	For instance, direct calculation (see for example \cite[Section 2.2]{CJJ11}) shows that $\beta^*(r) = \frac{1}{2} - r$ for $G_n = \delta_{n^{-r}}$, but $\beta^*(r) = \frac{1}{2} - 2 r$ for ${G}_n = \frac{1}{2}(\delta_{n^{-r}}+\delta_{-n^{-r}})$.



Moreover, unlike in the very sparse regime, moment-based tests can be powerful in the moderately sparse regime, which guarantee that $\betau \geq \frac{1}{2}$. For instance, in the above examples $G_n = \delta_{n^{-r}}$ or ${G}_n = \frac{1}{2}(\delta_{n^{-r}}+\delta_{-n^{-r}})$, the detection boundary can be obtained by thresholding the sample mean or sample variance respectively. More sophisticated moment-based tests such as the excess kurtosis tests have been studied in the context of sparse mixtures \cite{JSDAF05}. It is unclear whether they are always optimal when $\beta < \frac{1}{2}$.

\subsection{Adaptive optimality of higher criticism tests}
	\label{sec:hc.discuss}
	While \prettyref{thm:HC} establishes the adaptive optimality of the higher criticism test in the very sparse regime $\beta > \frac{1}{2}$,  the optimality of the higher criticism test in the moderately sparse case $\beta < \frac{1}{2}$ remains an open question.
	 Note that in the classical setup \prettyref{eq:HT.IDJ}, it has been shown \cite{CJJ11} that the higher criticism test achieves adaptive optimality for $\beta \in [0,\frac{1}{2}]$ and $\mu_n = n^{-r}$. In this case since $\mu_n=o(1)$, we have $\alpha \equiv 0$ and \prettyref{thm:main} thus does not apply. It is possible to obtain a counterpart of \prettyref{thm:main} and an analogous expression for $\beta^*$ for the moderately sparse regime if one assumes a similar uniform approximation property of the log-likelihood ratio, for example, $\ell_n(u \sqrt{\log n}) = n^{-\alpha(u)+o(1)}$ for some function $\alpha$. Another interesting problem is to investigate the optimality of procedures introduced in \cite{JW07} based on \renyi divergence under the same setup of \prettyref{thm:HC}.

%
%


\section{Proofs}
	\label{sec:pf}
\subsection{Auxiliary results}
	\label{sec:pfaux}
	Laplace's method (see, \eg, \cite[Section 2.4]{Erdelyi56}) is a technique for analyzing the asymptotics of integrals of the form $\int \exp(M f) \diff \nu$ when $M$ is large. 
The proof of \prettyref{thm:main} uses the following first-order version of the Laplace's method. 
	Since we are only interested in the exponent (\ie, the leading term), we do not use saddle-point approximation in the usual Laplace's method and impose \emph{no} regularity conditions on the function $f$ except for the finiteness of the integral. Moreover, the exponent only depends on the \emph{essential supremum} of $f$ with respect to $\nu$, which is invariant if $f$ is modified on a $\nu$-negligible set.

	\begin{lemma}
Let $(X,\calF, \nu)$ be a measure space. 
	Let $F: X \times \reals_+ \to \reals_+$ be measurable.
	Assume that
	\begin{equation}
	\lim_{M \diverge} \frac{\log F(x,M)}{M} = f(x)
	\label{eq:mainambda}
\end{equation}
holds uniformly in $x \in X$ for some measurable $f: X \to \reals$. 
	If $\int_{X} \exp(M_0 f) \diff \nu < \infty$ for some $M_0 > 0$, then 
	\begin{equation}
	\lim_{M \diverge} \frac{1}{M} \log \int_{X} F(x,M) \diff \nu = \esssup_{x\in X} f(x).
	\label{eq:laplace}
\end{equation}
	\label{lmm:laplace}
\end{lemma}

\vspace{-20pt}
\begin{proof}
First we deal with the case of $\esssup f = \infty$, which implies that $\nu(\{f>a\}) > 0$ for all $a > 0$. Moreover, by Chernoff bound, $\nu(\{f>a\}) < \exp(-M_0 a) \int \exp(M_0 f) \diff \nu < \infty$. By \prettyref{eq:mainambda}, for any $\epsilon > 0$, there exists $K>M_0$ such that
\begin{equation}
\exp(M(f(x) - \epsilon)) \leq F(x,M) \leq \exp(M(f(x) + \epsilon))	
	\label{eq:Funif}
\end{equation}
for all $x \in X$ and $M \geq K$. Therefore, $\int F(x,M) \diff \nu \geq \exp(-M \epsilon) \int \exp(M f) \diff \nu \geq \exp(M (a-\epsilon)) \nu(\{f>a\})$ for any $M > 0$ and $a > 0$. Then $\liminf_{M \diverge} \frac{1}{M} \log \int \exp(M f) \diff \nu \geq a-\epsilon$. By the arbitrariness of $a$ and $\epsilon$, we have
 $\lim_{M \diverge} \frac{1}{M} \log \int \exp(M f) \diff \nu =\infty$. 

Next we assume that $\esssup f < \infty$. By replacing $f$ with $f - \esssup f$, we can assume that $\esssup f = 0$ without loss of any generality. Then $f \leq 0$ $\nu$-a.e.  Hence, by \prettyref{eq:Funif},
\[
\int F(x,M) \diff \nu \leq \int \exp(M (f+\epsilon)) \diff \nu \leq \exp(M \epsilon) \int \exp(M_0 f) \diff \nu < \infty
\]
holds for all $M \geq K$. By the arbitrariness of $\epsilon$, we have
\[
\limsup_{M \diverge} \frac{1}{M} \log \int \exp(M f) \diff \nu \leq 0.
\] 
For the lower bound, note that, by the definition of $\esssup f = 0$, $\nu(\{f>-\delta\}) > 0$ for all $\delta > 0$. Therefore, by \prettyref{eq:Funif}, we have
\[
\int F(x,M) \diff \nu \geq \exp(-M \epsilon) \int  \exp(M f) \diff \nu \geq \exp(-M (\delta+\epsilon)) \nu(\{f>-\delta\})
\]
for any $M > 0$ and $\delta > 0$. First sending $M \diverge$ then $\delta \downarrow 0$ and $\epsilon \downarrow 0$, we have 
\[
\liminf_{M \diverge} \frac{1}{M} \log \int \exp(M f) \diff \nu \geq 0,
\]
completing the proof of \prettyref{eq:laplace}.
\end{proof}

The following lemma is useful for analyzing the asymptotics of Hellinger distance:
\begin{lemma}
\begin{enumerate}
	\item \label{sqrt1} For any $b > 0$, the function $s \mapsto (\sqrt{1 + b (s - 1)} - 1)^2$ is strictly convex on $\reals_+$ and strictly decreasing and increasing on $[0,1]$ and $[1,\infty)$, respectively. 
	\item \label{sqrt2} For any $t \geq 0$,
\begin{equation}
(\sqrt{2}-1)^2 t\wedge t^2	\leq (\sqrt{1+t}-1)^2 \leq t\wedge t^2.
	\label{eq:sqrt}
\end{equation}	
\end{enumerate}
	\label{lmm:sqrt}
\end{lemma}
\begin{proof}
\begin{enumerate}
	\item Since $t \mapsto \sqrt{1+t}$ is strictly concave, $s \mapsto (\sqrt{1 + b (s - 1)} - 1)^2 = 2 + b(s-1) - 2 \sqrt{b(s-1)}$ is strictly convex. Solving for the stationary point yields the minimum at $s=1$.
	
	\item First we consider $t \geq 1$. Since	$t \mapsto (\sqrt{1+t}-1)^2 = t - 2 \sqrt{1+t}$ is convex, $t \mapsto \frac{(\sqrt{1+t}-1)^2}{t}$ is increasing.
Consequently, we have $(\sqrt{2}-1)^2 \leq \frac{(\sqrt{1+t}-1)^2}{t} \leq 1$ for all $t \in [1,\infty)$.

Next we consider $0 \leq t \leq 1$. By the concavity of $t \mapsto \sqrt{1+t}$, $t \mapsto \frac{\sqrt{1+t}-1}{t}$ is  decreasing. Hence $\sqrt{2}-1 \leq \frac{\sqrt{1+t}-1}{t} \leq \frac{1}{2}$ for all $t \in [0,1]$. Assembling the above two cases yields \prettyref{eq:sqrt}. \qedhere
\end{enumerate}
\end{proof}
	
	The following lemmas are useful in proving \prettyref{thm:HC}:
	\begin{lemma}
Let $f: \reals \to \reals$ be measurable and $\mu$ be any measure on $\reals$. The function $g$ defined by
\[
g(s) = \esssup_{q \geq s} f(q)
\]
is decreasing and lower-semicontinuous, where the essential supremum is with respect to $\mu$.
	\label{lmm:gs}
\end{lemma}
\begin{proof}
	The monotonicity is obvious. We only prove lower-semicontinuity, which, in particular, also implies right-continuity. Let $s_n \to s$. By definition of the essential supremum, for any $\delta$, we have $\mu\{q \geq s: f(q) > g(s)-\delta\} > 0$. By the dominated convergence theorem, $\mu\{q \geq s_n: f(q) > g(s)-\delta\} \to \mu\{q \geq s: f(q) > g(s)-\delta\}$. Hence there exists $N$ such that $\mu\{q \geq s_n: f(q) > g(s)-\delta\} > 0$ for all $n \geq N$, which implies that $g(s_n) \geq g(s) - \delta$ for all $n \geq N$. By the arbitrariness of $\delta$, we have $\liminf_{n \diverge} g(s_n) \geq g(s)$, completing the proof of the lower semi-continuity.
\end{proof}

\begin{lemma}
Under the conditions of \prettyref{thm:main}, for any $u \geq 0$, 
	\[
\lim_{n \diverge} \frac{\log ((1-F_n(u \sqrt{2\log n})) \wedge F_n(-u \sqrt{2\log n}))}{\log n} = v(u) \triangleq \esssup_{q \geq u} \{\alpha(q) - q^2\}.
	\]	
	\label{lmm:Gntail}
\end{lemma}

\vspace{-15pt}
\begin{proof}
First assume that $u>0$. Then
	\begin{align*}
	1-F_n(u \sqrt{2\log n})
= & ~ \int_{y \geq u \sqrt{2\log n}}  \exp(\ell_n(y)) \phi(y) \diff y	
= \sqrt{\frac{\log n}{\pi}} \int_{q \geq u }  \exp(\ell_n(q \sqrt{2 \log n})) n^{-q^2} \diff q	\\
= & ~ n^{v(u)+o(1)},
\end{align*}
where the last equality follows from  \prettyref{lmm:laplace}. The proof for $u<0$ is completely analogous.
\end{proof}

\subsection{Proofs in \prettyref{sec:main}}
\label{sec:pfmain}

\begin{proof}[Proof of \prettyref{thm:main}]
Let $W \sim \calN(0,1)$. Put $\nu_n = (1-n^{-\beta}) \calN(0,1) + n^{-\beta} G_n$. Since $G_n \ll \Phi$ by assumption, we also have $\nu_n \ll \calN(0,1)$. Denote the likelihood ratio by $L_n = \frac{g_n}{\varphi} = \exp(\ell_n)$. Then 
\begin{equation}
\fracd{\nu_n}{\Phi} = 1  + n^{-\beta}(\exp(\ell_n) - 1). 	
	\label{eq:LRq}
\end{equation}

\noindent
\emph{(Direct part)} Recall the notation $\betas$ defined in \prettyref{eq:main}, which can be equivalently written as 
\[
\betas = \frac{1}{2} + \esssup_{u \in \reals} \sth{\alpha_+(u)- u^2 + \frac{u^2 \wedge 1}{2}}.
\]
Assuming \prettyref{eq:upn}, we show that $\betal \geq \beta^{\sharp}$ by lower bounding the Hellinger distance. To this end, fix an arbitrary $\delta > 0$. Let $\beta = \betas - 2 \delta$.	
Denote by $\lambda$ the Lebesgue measure on the real line. 
 By definition of the essential supremum, 
$
\lambda\{u: \alpha_+(u)- u^2 + \frac{u^2 \wedge 1}{2} \geq \beta + \delta - \frac{1}{2}\} > 0.		
$
Since $- u^2 + \frac{u^2 \wedge 1}{2} \leq 0$ for all $u$ and $\beta + \delta - \frac{1}{2} \geq - \delta$, we must have
$
\lambda\{u: \alpha(u)- u^2 + \frac{u^2 \wedge 1}{2} \geq \beta + \delta - \frac{1}{2}, \alpha(u) \geq 0\} > 0.		
$
Since, by assumption, $\lambda\{u: \alpha(u) > 0\} > 0$, there exists $0< \epsilon \leq \frac{\delta}{2}$, such that
\begin{equation}
\lambda\sth{u: \alpha(u)- u^2 + \frac{u^2 \wedge 1}{2} \geq \beta + \delta - \frac{1}{2}, \alpha(u) \geq 2 \epsilon} > 0.		
	\label{eq:case12}
\end{equation}
By assumption \prettyref{eq:upn}, there exists $N_{\epsilon} \in \naturals$ such that
\begin{equation}
\ell_{n}(u \sqrt{2 \log n}) \geq (\alpha(u)	- \epsilon) \log n
	\label{eq:upn1}
\end{equation}
holds for all $u \in \reals$ and all $n \geq N_{\epsilon}$.
From \prettyref{eq:case12}, we have either
\begin{equation}
\lambda\sth{u: |u| \leq 1, \alpha(u)- \frac{u^2}{2} \geq \beta + \delta - \frac{1}{2}, \alpha(u) \geq 2 \epsilon} > 0
	\label{eq:case1}
\end{equation}
or
\begin{equation}
\lambda\sth{u: |u| \geq 1, \alpha(u)- u^2 \geq \beta + \delta - 1, \alpha(u) \geq 2 \epsilon} > 0.
	\label{eq:case2}
\end{equation}
Next we discuss these two cases separately:

\noindent
\emph{Case I}: Assume \prettyref{eq:case1}. Let 
\begin{equation}
U = \frac{W}{\sqrt{2 \log n}} \sim \calN\pth{0,\frac{1}{2\log n}}.	
	\label{eq:UW}
\end{equation}
The square Hellinger distance can be lower bounded as follows:
\begin{align}
H_n^2(\beta) 
= & ~ H^2(P, \nu_n)	
=  \int \pth{\sqrt{\fracd{\nu_n}{P}}-1}^2 \diff P \nonumber \\
= & ~ \expect{\pth{\sqrt{1 + n^{-\beta} (\exp(\ell_n(U \sqrt{2 \log n})) - 1)} - 1}^2} \label{eq:acc0}\\
\geq & ~ \expect{\pth{\sqrt{1 + n^{-\beta} (\exp(\ell_n(U \sqrt{2 \log n})) - 1)} - 1}^2 \indc{|U| \leq 1, \alpha(U)- \frac{U^2}{2} \geq \beta + \delta - \frac{1}{2}, \alpha(U) \geq 2 \epsilon}} \nonumber \\
\geq & ~ \expect{\pth{\sqrt{1 + n^{-\beta} (n^{\alpha(U) - \epsilon} - 1)}-1}^2 \indc{|U| \leq 1, \alpha(U)- \frac{U^2}{2} \geq \beta + \delta - \frac{1}{2}, \alpha(U) \geq 2 \epsilon}} \label{eq:acc2}\\
\geq & ~ \frac{(\sqrt{2}-1)^2}{4} \expect{ n^{(\alpha(U) - \epsilon-\beta) \wedge 2(\alpha(U) - \epsilon-\beta)} \indc{|U| \leq 1, \alpha(U)- \frac{U^2}{2} \geq \beta + \delta - \frac{1}{2}, \alpha(U) \geq 2 \epsilon}} \label{eq:acc3}\\
= & ~ \frac{(\sqrt{2}-1)^2 \sqrt{\log n}}{4\sqrt{\pi}} \int n^{(\alpha(u) - \epsilon-\beta) \wedge 2(\alpha(u) - \epsilon-\beta) - u^2} \indc{|u| \leq 1, \alpha(u)- \frac{u^2}{2} \geq \beta + \delta - \frac{1}{2}, \alpha(u) \geq 2 \epsilon}  \diff u  \label{eq:acc4} \\
\geq & ~ \frac{(\sqrt{2}-1)^2 \sqrt{\log n}}{4\sqrt{\pi}} \lambda\sth{|u| \leq 1, \alpha(u)- \frac{u^2}{2} \geq \beta + \delta - \frac{1}{2}, \alpha(u) \geq 2 \epsilon} n^{-1+\frac{\delta}{2}}  \label{eq:acc5}
\end{align}
where 
\begin{itemize}
	\item \prettyref{eq:acc0}: By \prettyref{eq:LRq}.
	\item \prettyref{eq:acc2}: By \prettyref{lmm:sqrt}.\ref{sqrt1} and \prettyref{eq:upn1}.
	\item \prettyref{eq:acc3}: Without loss of generality, we can assume that $n^{\epsilon} \geq 2$. Then applying the lower bound in \prettyref{lmm:sqrt}.\ref{sqrt2} yields the desired inequality.
	\item \prettyref{eq:acc4}: We used the density of $U$ defined in \prettyref{eq:UW}.
	\item \prettyref{eq:acc5}: Given that $|u| \leq 1$ and $\alpha(u)- \frac{u^2}{2} \geq \beta + \delta - \frac{1}{2}$, we have both
	$\alpha(u) - \epsilon-\beta- u^2 \geq -\frac{1+v^2}{2}+\delta - \epsilon \geq -1+\frac{\delta}{2}$ and 
	$2\alpha(u) - 2 \epsilon-2 \beta - u^2 \geq -1+2\delta-2\epsilon \geq -1 + \delta$.
\end{itemize}

\noindent
\emph{Case II}:  Now we assume \prettyref{eq:case2}. Following analogous steps as in the previous case, we have
\begin{align}
H_n^2(\beta) 
\geq & ~ \frac{(\sqrt{2}-1)^2 \sqrt{\log n}}{4\sqrt{\pi}} \int n^{(\alpha(u) - \epsilon-\beta) \wedge 2(\alpha(u) - \epsilon-\beta) - u^2} \indc{|u| \geq 1, \alpha(u)- u^2 \geq \beta + \delta - 1, \alpha(u) \geq 2 \epsilon}  \diff u  \nonumber \\
\geq & ~ \frac{(\sqrt{2}-1)^2 \sqrt{\log n}}{4\sqrt{\pi}} \lambda\sth{|u| \geq 1, \alpha(u)- u^2 \geq \beta + \delta - 1, \alpha(u) \geq 2 \epsilon} n^{-1+\frac{\delta}{2}}  \label{eq:acc6}
\end{align}
where \prettyref{eq:acc6} is due to the following: Since $|u| \geq 1$ and $\alpha(u)- u^2 \geq \beta + \delta - 1$, we have both
	$\alpha(u) - \epsilon-\beta- u^2 \geq \delta - \epsilon-1 \geq -1+\frac{\delta}{2}$ and 
	$2\alpha(u) - 2 \epsilon-2 \beta - u^2 \geq v^2-2 +2\delta-2\epsilon \geq -1 + \delta$.

Combining \prettyref{eq:acc5} and  \prettyref{eq:acc6} we conclude that $H_n^2(\beta) = \omega(n^{-1})$. By the arbitrariness of $\delta > 0$ and the alternative definition of $\betal$ in \prettyref{eq:betalH}, the proof of $\betal \geq \beta^{\sharp}$ is completed.

\noindent
\emph{(Converse part)} 
Fix an arbitrary $\delta > 0$. Let
\begin{equation}
	\beta = \betas + 2 \delta.	
	\label{eq:convbeta}
\end{equation}
We upper bound the Hellinger integral as follows: First note that
\begin{equation}
H_n^2(\beta) = \expect{\pth{\sqrt{1 + n^{-\beta} (L_n - 1)} - 1}^2 \indc{L_n \geq 1}} + \expect{\pth{\sqrt{1 + n^{-\beta} (L_n - 1)} - 1}^2 \indc{L_n \leq 1}}.	
	\label{eq:H3}
\end{equation}
Applying \prettyref{lmm:sqrt}.\ref{sqrt1}, we have
\begin{equation}
\expect{\pth{\sqrt{1 + n^{-\beta} (L_n - 1)} - 1}^2 \indc{L_n \leq 1}} \leq (\sqrt{1-n^{-\beta}}-1)^2 \leq n^{-2\beta} = o(n^{-1}),	
	\label{eq:H33}
\end{equation}
since $\beta > \betas \geq \frac{1}{2}$ by \prettyref{eq:convbeta}.
Consequently, the asymptotics of the Hellinger integral $H_n^2(\beta)$ is dominated by the first term in \prettyref{eq:H3}, denoted by $a_n$, which we analyze below using the Laplace method. 

By \prettyref{eq:lown}, there exists $N_{\delta} \in \naturals$ such that
\begin{equation}
\ell_{n}(u \sqrt{2 \log n}) \leq (\alpha(u)	+ \delta) \log n
	\label{eq:upn2}
\end{equation}
holds for all $u \in \reals$ and all $n \geq N_{\delta}$. Then
 \begin{align}
a_n
\triangleq & ~ \expect{\pth{\sqrt{1 + n^{-\beta} (L_n - 1)} - 1}^2 \indc{L_n \geq 1}}	\nonumber \\
= & ~ \expect{\pth{\sqrt{1 + n^{-\beta} (\exp(\ell_n(U \sqrt{2 \log n})) - 1)} - 1}^2 \indc{L_n \geq 1}} \nonumber \\
\leq & ~ \expect{\pth{\sqrt{1 + n^{-\beta} (n^{\alpha(U) + \delta} - 1)} - 1}^2 \indc{\alpha(U) \geq -\delta}} \label{eq:cv0}\\
\leq & ~ \expect{\pth{\sqrt{1 + n^{\alpha(U) + \delta-\beta}} - 1}^2 } \nonumber \\
\leq & ~ \expect{n^{(2(\alpha(U) + \delta-\beta)) \wedge (\alpha(U) + \delta-\beta) }} \label{eq:cv1}\\
= & ~ \sqrt{\frac{\log n}{\pi}} \int n^{(2(\alpha(u) + \delta-\beta))\wedge (\alpha(u) + \delta-\beta) - u^2} \diff u \label{eq:cv2}
\end{align}
where \prettyref{eq:cv0} and \prettyref{eq:cv1} are due to \prettyref{eq:upn2} and \prettyref{lmm:sqrt}.\ref{sqrt2}, respectively. 
Next we apply \prettyref{lmm:laplace} to analyze the exponent of \prettyref{eq:cv2}. First we verify the integrability condition:
 \[
 \int n^{(2(\alpha(u) + \delta-\beta))\wedge (\alpha(u) + \delta-\beta) - u^2} \diff u \leq n^{\delta-\beta} \int n^{\alpha(u) - u^2} \diff u 
 < \infty
 \]
 in view of \prettyref{eq:alpha.legit}.
 Applying \prettyref{eq:laplace} to \prettyref{eq:cv2}, we have
 \begin{align}
a_n
\leq & ~ n^{\esssup_{u \in \reals} \{(2(\alpha(u) + \delta-\beta))\wedge (\alpha(u) + \delta-\beta) - u^2 \} + o(1)} .
\label{eq:cv3}
\end{align}
By \prettyref{eq:convbeta}, $\alpha(u)- u^2 + \frac{u^2 \wedge 1}{2} \leq \beta - \frac{1}{2} - 2 \delta$ holds a.e. Consequently, $\alpha(u)- u^2 \leq \beta - 1- 2 \delta$ holds for almost every $u \in (-\infty,-1] \cup [1,\infty)$ and $\alpha(u)- \frac{u^2}{2} \leq \beta - \frac{1}{2} - 2 \delta$ holds  for almost every $u \in [-1,1]$. These conditions immediately imply that 
 \begin{equation}
(2(\alpha(u) + \delta-\beta))\wedge (\alpha(u) + \delta-\beta) - u^2 \leq -1 + \delta	
	\label{eq:cv4}
\end{equation}
holds a.e. Assembling \prettyref{eq:H3} and \prettyref{eq:cv3}, we conclude that $H_n^2(\beta) = o(n^{-1})$. By the arbitrariness of $\delta > 0$ and the alternative definition of $\betau$ in \prettyref{eq:betauH}, the proof of $\betau \leq \beta^{\sharp}$ is completed. 
\end{proof}

\begin{proof}[Proof of \prettyref{thm:Ebeta}]
In view of the proof of \prettyref{thm:main}, the desired \prettyref{eq:Hnexp} readily follows from combining \prettyref{eq:acc5}, \prettyref{eq:acc6}, \prettyref{eq:H33} and \prettyref{eq:cv3}.
\end{proof}

\begin{proof}[Proof of \prettyref{lmm:alpha}]
	Put
	\begin{equation}
	c(t) \triangleq \int \exp(t (\alpha(u) - u^2)) \diff u.
	\label{eq:clambda}
\end{equation}

\noindent
\emph{(Necessity)}
Since $c(t) \geq 0$, it is sufficient to prove 
\begin{equation}
\limsup_{t \diverge} \frac{\log c(t)}{t} \leq 0.
	\label{eq:alpha.legit1}
\end{equation}
Since $\int g_n = 1$, we have $\int g_n(u \sqrt{\log n}) \diff u = (\log n)^{-\frac{1}{2}}$. 
 By assumption, $g_n(u\sqrt{\log n}) = n^{\alpha(u)-u^2+o(1)}$ uniformly in $u$. Then for all $\delta > 0$, $c(\log n) = \int n^{\alpha(u)-u^2} \diff u \leq \frac{n^{\delta}}{\sqrt{\log n}} < \infty$ holds for sufficiently large $n$. In particular, $c(\log n) \leq n^{o(1)}$. For general $t > 0$, let $n_1 = \floor{\exp(t)}$, $n_2 = \ceil{\exp(t+1)}$ and $t_i = \log n_i, i=1,2$. Put $p = \frac{t_2-t_1}{t-t_1}, q = \frac{t_2-t_1}{t_2-t}, a = \frac{t_1}{p}$ and $b = \frac{t_1}{q}$. Then $\frac{1}{p},\frac{1}{q} \in [0,1]$. H\"older's inequality yields $c(t) = \int \eexp^{a(\alpha(u)-u^2)} \eexp^{b(\alpha(u)-u^2)} \diff u \leq c(t_1)^{\frac{1}{p}} c(t_2)^{\frac{1}{q}} \leq c(\log n_1) c(\log n_2) \leq \exp(o(t))$, which gives the desired \prettyref{eq:alpha.legit1}. It then follows from \prettyref{lmm:laplace} that $\esssup_u \{\alpha(u) -u^2\} \leq 0$, \ie, $\alpha(u) \leq u^2$ a.e.
 
 \noindent
\emph{(Sufficiency)} 
Let $\alpha$ be a measurable function satisfying \prettyref{eq:alpha.legit}. Let $G_n$ be a probability measure with the density
\[
g_n(y) = \frac{1}{c(\log n) \sqrt{\log n} } \exp\sth{\alpha\pth{\frac{y}{\sqrt{2\log n}}} \log n -\frac{y^2}{2}},
\]
which is a legitimate density function in view of \prettyref{eq:clambda}. Then the log-likelihood ratio satisfies $\ell_n(u\sqrt{\log n}) =  \log \frac{\sqrt{2 \pi}}{c(\log n) \sqrt{\log n}} + \alpha(u)$, which fulfills \prettyref{eq:alphau} uniformly.

For convolutional models, the convexity of $\alpha$ is inherited from the geometric properties of the log-likelihood ratio in the normal location model: Since $y \mapsto \log \frac{\expect{\varphi(y-X)}}{\varphi(y)}$ is convex for any random variable $X$ (see, \eg, \cite[Property 3]{hatsell} and \cite{esposito}), we have $\ell_n(((1-t) u + t v)\sqrt{2 \log n}) \leq (1-t) \ell_n(u \sqrt{2 \log n}) + t \ell_n(v \sqrt{2 \log n})$ for any $t \in [0,1]$ and $u,v\in \reals$. Dividing both sides by $\log n$ and sending $n \diverge$, we have $\alpha((1-t) u + t v) \leq (1-t) \alpha(u) + t \alpha(v)$.
\end{proof}

\begin{proof}[Proof of \prettyref{cor:conv}]
Since $g_n = \varphi * p_n$, we have
\begin{align*}
g_n(u \sqrt{2 \log n})
= & ~ \int_{\reals} \varphi(u \sqrt{2 \log n}  - x ) p_n(x) \diff x	\\
= & ~  \sqrt{2 \log n} \int_{\reals} \varphi((u-t) \sqrt{2 \log n}   ) p_n(x \sqrt{2 \log n} ) \diff x	\\
= & ~ n^{o(1)} \int_{\reals} n^{- (u - t)^2 -f(t)+o(1)} \diff x	\\
= & ~ n^{-\essinf_{z \in \reals} \{(u - t )^2 + f(t)\} + o(1) }
\end{align*}
where the last equality follows from \prettyref{lmm:laplace}. Plugging the above asymptotics into $\ell_n = \log \frac{g_n}{\varphi}$, we see that \prettyref{eq:alphau} is fulfilled uniformly in $u\in \reals$ with
$
\alpha(u) = u^2 - \essinf_{z \in \reals} \{(u - \sqrt{r} z )^2 + |z|^{\tau}\}.
$
Applying \prettyref{thm:main}, we obtain
\begin{align*}
\beta^*
= & ~ \frac{1}{2} + \esssup_{u \in \reals} \esssup_{t \in \reals} \sth{ -(u - t )^2 -f(t) + \frac{u^2 \wedge 1}{2} } 	\\
= & ~ \frac{1}{2} + \esssup_{t \in \reals} \sth{  -f(t) + \esssup_{u \in \reals} \sth{-(u - t )^2 +  \frac{u^2 \wedge 1}{2} }} 	\\
= & ~ \sup_{t \in \reals} \{\beta^*_{\rm IDJ}(t^2) - f(t) \}
\end{align*}
where the last step follows from the \prettyref{eq:betaIDJ1}.
\end{proof}

\begin{proof}[Proof of \prettyref{thm:nG}]
	Let $W_n \sim Q_n$. Put $\nu_n = (1-n^{-\beta}) \Phi + n^{-\beta} G_n$. Since $G_n \ll Q_n$ by assumption, we also have $\nu_n \ll P$. Denote the likelihood ratio (Radon-Nikodym derivative) by $L_n = \frac{\diff G_n}{\diff Q_n} = \exp(\ell_n)$. Then 
\begin{equation}
\fracd{\nu_n}{Q_n} = 1  + n^{-\beta}(\exp(\ell_n) - 1). 	
	\label{eq:dQnP}
\end{equation}

Instead of introducing the random variable $U$ in \prettyref{eq:UW} for the Gaussian case, we apply the quantile transformation to generate the distribution of $W_n$: Let $U$ be uniformly distributed on the unit interval. Then $S = \log \frac{1}{U}$ which is exponentially distributed. Putting
$
S_n = \frac{S}{\log n},
$
we have
\begin{equation}
	W_n \eqdistr z_n(U) = z_n\pth{n^{-S_n}} \eqdistr z_n\pth{1-n^{-S_n}}.
	\label{eq:WSn}
\end{equation}
Set
$r_n(s) =  \ell_n \circ z_n (n^{-s}) $ and	
$t_n(s) =  \ell_n \circ z_n (1-n^{-s})$,
which satisfy
\begin{align}
\sup_{s \geq \log_n 2} |r_n(s) - \alpha_0(s) \log n|  \leq & ~ \delta \log n 	\label{eq:rnsa} \\
\sup_{s \geq \log_n 2} |t_n(s) - \alpha_1(s)	\log n| 	\leq & ~ \delta \log n	\label{eq:tnsa}
\end{align}
for all sufficiently large $n$. For the converse proof, we can write the square Hellinger distance as an expectation with respect to $S_n$:
\begin{align*}
H_n^2(\beta)
= & ~ \expect{\pth{\sqrt{1 + n^{-\beta} (\exp(\ell_n(z_n(U))) - 1)} - 1}^2 \indc{0 < U < \frac{1}{2}}}  \\
	& ~ + \expect{\pth{\sqrt{1 + n^{-\beta} (\exp(\ell_n(z_n(1-U))) - 1)} - 1}^2 \indc{0 < U \leq \frac{1}{2}}} .
\end{align*}
Analogous to \prettyref{eq:H33}, by truncating the log-likelihood ratio at zero, we can show that the Hellinger distance is dominated by the following:
\begin{align}
a_n
= & ~ \expect{\pth{\sqrt{1 + n^{-\beta} (\exp(\ell_n(z_n(U))) - 1)} - 1}^2 \indc{0 < U < \frac{1}{2}, r_n(S_n) \geq 0}} \nonumber   \\
	& ~ + \expect{\pth{\sqrt{1 + n^{-\beta} (\exp(\ell_n(z_n(1-U))) - 1)} - 1}^2 \indc{0 < U \leq \frac{1}{2}, t_n(S_n) \geq 0}} \nonumber \\
= & ~ \expect{\pth{\sqrt{1 + n^{-\beta} (\exp(r_n(S_n)) - 1)} - 1}^2 \indc{S_n > \log_n 2, r_n(S_n) \geq 0}} \label{eq:jg1} \\
	& ~ + \expect{\pth{\sqrt{1 + n^{-\beta} (\exp(t_n(S_n)) - 1)} - 1}^2 \indc{S_n \geq  \log_n 2, t_n(S_n) \geq 0}} \nonumber \\
\leq & ~ \expect{\pth{\sqrt{1 + n^{-\beta} (n^{\alpha_0(S_n)+\delta} - 1)} - 1}^2 + \pth{\sqrt{1 + n^{-\beta} (n^{\alpha_1(S_n)+\delta} - 1)} - 1}^2 } \label{eq:jg2} \\
\leq & ~ 2 \, \expect{n^{2(\alpha_0 \vee \alpha_1(U) + \delta-\beta) \wedge (\alpha_0 \vee \alpha_1(U) + \delta-\beta) }} \\
\leq & ~ n^{-1-\delta} \label{eq:jg3} 
\end{align}
where \prettyref{eq:jg1} follows from \prettyref{eq:dQnP} -- \prettyref{eq:WSn}, \prettyref{eq:jg2} from \prettyref{eq:rnsa} -- \prettyref{eq:tnsa} and \prettyref{eq:jg3} from \prettyref{eq:cv2} -- \prettyref{eq:cv4}. The direct part of the proof is completely analogous to that of \prettyref{thm:main} by lower bounding the integral in \prettyref{eq:jg1}.
\end{proof}

\subsection{Proof of \prettyref{thm:HC}}
	\label{sec:pfHC}
\begin{proof}
Let $U_i = \Phi(X_i)$, which is uniformly distributed on $[0,1]$ under the null hypothesis. With a change of variable, we have
\begin{align}
\HC_n 
= & ~ \sqrt{n} 	\sup_{t \in \reals} \frac{|\bbF_n(t) - \Phi(t)|}{\sqrt{\Phi(t) \bar{\Phi}(t)}} \label{eq:HCtest1}\\
= & ~ \sqrt{n} 	\sup_{0<u<1} \frac{|\bbF_n(\Phi^{-1}(u)) - u|}{\sqrt{u (1-u)}},
\end{align}
which satisfies that $\frac{\HC_n}{\sqrt{2 \log \log n}} \toprob 1$ \cite[p. 604]{Shorack.Wellner}. Therefore the Type-I error probability of the test \prettyref{eq:HCtest} vanishes for any choice of $\delta > 0$. It remains to show that $\HC_n = \omega_{\Prob}(\log\log n)$ under the alternative.
To this end, fix $0 < s < 1$ and put $r_{n,s}=\Phi(\sqrt{2s\log n})$ and $\rho_{n,s} = (1-n^{-\beta}) \Phi(\sqrt{2s\log n}) + n^{-\beta} {G_n}(\sqrt{2s\log n})$. 
By \prettyref{eq:HCtest1}, we have
\begin{align}
\HC_n \geq V_n(s)
\triangleq & ~ 	\sqrt{n} \frac{\bbF_n(\sqrt{2s\log n}) - r_{n,s}}{\sqrt{r_{n,s}(1-r_{n,s})}} \label{eq:HCVn} \\
= & ~ 	\frac{N_n(s) - n r_{n,s}}{\sqrt{n r_{n,s}(1-r_{n,s})}},
\end{align}
where $N_n(s) \triangleq \sum_{i=1}^n \indc{X_i \geq \sqrt{2s\log n}}$ is binomially distributed with sample size $n$ and success probability $\rho_{n,s}$. Therefore
\begin{equation}
\expect{V_n(s)} 
= \sqrt{n}
\frac{\rho_{n,s} - r_{n,s}}{\sqrt{r_{n,s}(1-r_{n,s})}} 
= n^{\frac{1}{2}-\beta} \frac{{G_n}(\sqrt{2s\log n}) - r_{n,s}}{\sqrt{r_{n,s}(1-r_{n,s})}}	.
	\label{eq:EVn}
\end{equation}
and
\begin{equation}
\var V_n(s) = 
\frac{\rho_{n,s}(1-\rho_{n,s})}{r_{n,s}(1-r_{n,s})}.
	\label{eq:varVn}
\end{equation}
By Chebyshev's inequality, 
\[
\prob{V_n(s) \leq \frac{1}{2} \expect{V_n(s)}} \leq \frac{4 \, \var V_n(s)}{\expect{V_n(s)}^2} = \frac{4 \rho_{n,s}(1- \rho_{n,s})}{n (\rho_{n,s} - r_{n,s})^2}.
\]
By \prettyref{lmm:Gntail}, 
\begin{equation}
1-G_n(\sqrt{2s\log n}) = n^{v(s)+o(1)},	
	\label{eq:Gnexp}
\end{equation}
where $v(s) = \esssup_{q \geq s} \{\alpha(q) - q\} \geq -s$. Plugging \prettyref{eq:Gnexp} into \prettyref{eq:EVn} and \prettyref{eq:varVn} yields
\begin{equation}
\expect{V_n(s)} = n^{\frac{1+s}{2}-\beta+v(s)+o(1)}	
	\label{eq:EVns}
\end{equation}
and
\begin{equation}
\prob{V_n(s) \leq \frac{1}{2} \expect{V_n(s)}} \leq n^{2\beta-s-1-2v(s)+o(1)}+ n^{\beta-1-v(s)+o(1)}.	
	\label{eq:PVns}
\end{equation}
Suppose that $\beta < \frac{1+s}{2}+v(s)$. Then $\expect{V_n(s)} = \omega(\sqrt{\log \log n})$. Moreover, we have $2\beta-s-1-2v(s)<0$ and $\beta-1-v(s) \leq \frac{s-1}{2} < 0$ since $s<1$. Combining \prettyref{eq:HCVn}, \prettyref{eq:EVns} and \prettyref{eq:PVns}, we obtain
\[
\prob{\HC_n > \sqrt{(2+\delta)\log \log n}} = 1-o(1),
\]
that is, the Type-II error probability also vanishes. Consequently, a sufficient condition for the higher criticism test to succeed is
\begin{align}
\beta 
< & ~ \sup_{0<s<1} \frac{1+s}{2}+v(s) \label{eq:betavs}\\
= & ~ \esssup_{0<s<1} \frac{1+s}{2}+v(s), \label{eq:betavs1}
\end{align}
where \prettyref{eq:betavs1} follows from the following reasoning: By \cite[Proposition 3.5]{PH96}, the supremum and the essential supremum (with respect to the Lebesgue measure) coincide for all lower semi-continuous functions. Indeed, $v$ is lower semi-continuous by \prettyref{lmm:gs}, and so is $s \mapsto \frac{1+s}{2}+v(s)$.

It remains to show that the right-hand side of \prettyref{eq:betavs1} coincides with the expression of $\beta^*$ in \prettyref{thm:main}. Indeed, we have
\vspace{-15pt}
\begin{align*}
	\esssup_{0 \leq s \leq 1} \sth{s + 2 v(s) }	
= & ~\esssup_{0 \leq s \leq 1} \sth{s + 2 \esssup_{q \geq s} \{\alpha(q)-q \} }	\\
= & ~ \esssup_{q \geq 0} \esssup_{q \wedge 1 \leq s \leq 1} \sth{2 \alpha(q)-2q + s  } 	\\
= & ~ \esssup_{q \geq 0} \sth{2 \alpha(q)-2q + q \wedge 1}
\end{align*}
Note that the second equality follows from interchanging the essential supremums: For any bi-measurable function $(x,y)\mapsto f(x,y)$,
\[
\esssup_x \esssup_y f(x,y) = \esssup_y \esssup_x f(x,y) = \esssup_{x,y} f(x,y),
\]
where the last essential supremum is with respect to the product measure. Thus the proof of the theorem is completed.
\end{proof}

\appendix

\section{Hellinger distances for mixtures}	
\label{app:H2mix}
This appendix collects a few properties of total variation and Hellinger distances	for mixture distributions.	
	\begin{lemma}
Let $0 \leq \epsilon \leq 1$ and $Q_1 \perp P$. Then
\begin{equation}
H^2(P,(1-\epsilon)Q_0+\epsilon Q_1) = 2(1-\sqrt{1-\epsilon}) + \sqrt{1-\epsilon} \,  H^2(P,Q_0)
	\label{eq:H2.mix}
\end{equation}
which satisfies
\begin{equation}
	\frac{1}{4} \leq \frac{H^2(P,(1-\epsilon)Q_0+\epsilon Q_1)}{\epsilon \vee H^2(P,Q_0)} \leq 4
	\label{eq:H2.mix1}
\end{equation}

	\label{lmm:H2.mix}
\end{lemma}
\begin{proof}
Since $Q_1 \perp P$, there exists a measurable set $E$ such that $P(E) = 0$ and $Q_1(E)=1$. 
Then
\begin{align*}
H^2(P,(1-\epsilon)Q_0+\epsilon Q_1)
= & ~ 2 - 2 \int \sqrt{\diff P ((1-\epsilon) \diff Q_0+\epsilon \diff Q_1)}\\
= & ~ 2 - 2 \sqrt{1-\epsilon} \int_{\comp{E}} \sqrt{\diff P \diff Q_0 }\\
= & ~ 2 - \sqrt{1-\epsilon} \, (2 - H^2(P,Q_0)).
\end{align*}
The inequalities in \prettyref{eq:H2.mix1} follow from \prettyref{eq:H2.mix} and the facts that $\frac{\epsilon}{2} \leq \sqrt{1-\epsilon} \leq \epsilon$ and $0 \leq H^2 \leq 2$.
\end{proof}
	
	\begin{lemma}
	For any probability measures $(P, Q)$, $\epsilon \mapsto H^2(P, (1-\epsilon)P+\epsilon Q)$  is decreasing on $[0,1]$.
	\label{lmm:Hdec}
\end{lemma}
\begin{proof}
Fix $0 \leq \epsilon < \epsilon' \leq 1$. Since $(1-{\epsilon})P+{\epsilon} Q = ((1-\epsilon')P+\epsilon' Q ) \frac{\epsilon}{\epsilon'} + \frac{\hat{\epsilon}-\epsilon}{\epsilon'} P$, the convexity of $H^2(P,\cdot)$ yields
\[
H^2((1-\epsilon) P + \epsilon Q, P) \leq \frac{\epsilon}{\epsilon'} H^2((1-\epsilon')P+\epsilon' Q, P).  \qedhere
\]
\end{proof}

%
%
%

	We conclude this appendix by proving \prettyref{lmm:beta01} presented in \prettyref{sec:FL}:
	\begin{proof}
By \prettyref{lmm:Hdec}, the function $\beta \mapsto H_n^2(\beta)$ is decreasing, which, in view of the characterization \prettyref{eq:betalH} -- \prettyref{eq:betauH}, implies that $\betal \leq \betau$. Thus it only remains to establish the rightmost inequality in \prettyref{eq:beta01}. To this end, we show that as soon as $\beta$ exceeds $1$, $V_n(\beta)$ becomes $o(1)$ regardless of the choice of $\{G_n\}$: Fix $\beta > 1$. Then
\begin{align}
V_n(\beta)
= & ~  \TV(\Phi^n, ((1-n^{-\beta})\Phi + n^{-\beta} G_n)^n) \nonumber \\
\leq & ~ 	 \TV(\delta_0^n, ((1-n^{-\beta})\delta_0 + n^{-\beta} \delta_1)^n) \label{eq:dp}\\
= & ~ 	1- (1-n^{-\beta})^n \nonumber \\
\leq & ~ 	n^{1-\beta} \nonumber \\
= & ~ 	o(1), \nonumber 
\end{align}
where \prettyref{eq:dp} follows from the data-processing inequality, which is satisfied for all $f$-divergences \cite{Csiszar67}, in particular, the total variation: $\TV(P_Y, Q_Y) \leq \TV(P_X,Q_X)$, where $Q_{Y|X} = P_{Y|X}$ is any probability transition kernel.
\end{proof}
	
\begin{remark}
While \prettyref{lmm:Hdec} is sufficient for our purpose in proving \prettyref{lmm:beta01}, it is unclear whether the monotonicity carries over to $\epsilon \mapsto \TV(P^n, ((1-\epsilon)P+\epsilon Q)^n)$, since product measures do not form a convex set. It is however easy to see that $\epsilon \mapsto \TV(P^n, ((1-\epsilon)P+\epsilon Q)^n)$ is decreasing, which follows from the proof of \prettyref{lmm:Hdec} with $H^2$ replaced by $\TV$. It is also clear that $\epsilon \mapsto H^2(P^n, ((1-\epsilon)P+\epsilon Q)^n)$ is decreasing in view of \prettyref{eq:hprod}.
	\label{rmk:Hdec}
\end{remark}

\section{The implication of the condition \prettyref{eq:tail}}
	\label{app:tail}
	In this appendix we show that \prettyref{eq:tail} implies that $\beta^*= 1$, \ie, for any $\beta < 1$, the hypotheses in \prettyref{eq:HT} can be tested reliably. Without loss of generality, we assume that $u \geq 1$. Then
	\[
	\tau_n \triangleq G_n((\sqrt{2\log n}, \infty)) = n^{-o(1)},
	\]
	We show that the total variation distance between the product measures converge to one. Put $A_n = (-\infty, \sqrt{2s \log n}]^n$. 
	In view of the first inequality in \prettyref{eq:TV}, the total variation distance can be lower bounded as follows:
	\[
V_n(\beta) \geq \Phi^n(A_n) - ((1-n^{-\beta})\Phi + n^{-\beta} G_n)^n(A_n).
	\]
	Using \prettyref{eq:phic}, we have
	\begin{align*}
\Phi^n(A_n)
= & ~ (1-\bar{\Phi}(\sqrt{2 s \log n}))^n	
= 1-\frac{n^{1-s}}{\sqrt{4 \pi s \log n}} (1+o(1)).	
\end{align*}
On the other hand,
	\begin{align*}
((1-n^{-\beta})\Phi + n^{-\beta} G_n)^n(A_n)
= & ~ (1-(1-n^{-\beta})\bar{\Phi}(\sqrt{2 s \log n}) - n^{-\beta} \tau_n )^n	\\
= & ~ (1-n^{-s+o(1)} - n^{-\beta-s+o(1)} - n^{-\beta+o(1)})^n	\\
= & ~ o(1)
\end{align*}
where the last equality is due to $0<\beta < 1 \leq s$. Therefore $V_n(\beta)=1-o(1)$ for any $\beta < 1$, which proves that $\beta^*= 1$.

In fact, the above derivation also shows that the following \emph{maximum test} achieves vanishing probability of error: declare $H_1$ if and only if $\max_i |X_i| > |u| \sqrt{2 \log n}$. In general the maximum test is suboptimal. For example, in the classical setting \prettyref{eq:HT.IDJ} where $G_n = \delta_{\mu_n}$, 
 \cite[Theorem 1.3]{DJ.HC} shows that the maximum test does not attain the Ingster-Donoho-Jin detection boundary for $\beta \in [\frac{1}{2},\frac{3}{4}]$.



\end{document}